\documentclass{article}

\usepackage{fullpage} 
\usepackage{mathptmx}
\usepackage{times}
\usepackage{subfig}

\usepackage{amsmath, amssymb, amsthm, amsfonts}
\usepackage{graphicx, algorithm, algorithmic, ifthen}

\usepackage{epsfig}
\usepackage{amssymb}
\usepackage{amsmath}
\usepackage{amsfonts}

\newtheorem{definition}{Definition}
\newtheorem{theorem}{Theorem}
\newtheorem{lemma}{Lemma}

\newtheorem{example}{Example}

\newtheorem{corollary}{Corollary}
\newtheorem{proposition}{Proposition}

\newcommand{\Flip}{{\tt\textsc{Flip}}}

\newcommand{\Out}{{\tt\textsc{Out}}}

\newcommand{\Worlds}{{\tt{Worlds}}}
\newcommand{\widebar}[1]{{\overline{#1}}}

\newcommand{\secureview}{{\tt Secure-View}~}  
\newcommand{\SAsecureview}{{\tt Secure-View}}  
\newcommand{\safeview}{{\tt Safe-View}~}

\newcommand{\scream}[1]{{\bf * #1 *}{\typeout{#1}}}

\newcommand{\eat}[1]{}
\newcommand{\neweat}[1]{}

\newcommand{\angb}[1]{\langle #1 \rangle}

\newcommand{\cost}{{\tt c}}
\newcommand{\pvt}{{\tt c}}

\newcommand{\be}{\begin{enumerate}}
\newcommand{\ee}{\end{enumerate}}
\newcommand{\Algo}{\textsc{Algo}}
\newcommand{\No}{\textsc{No}}
\newcommand{\Yes}{\textsc{Yes}}
\newcommand{\Dom}{\texttt{Dom}}
\newcommand{\Range}{\texttt{Range}}
\newcommand{\tup}[1]{\mathbf{#1}}
\newcommand{\proj}[2]{\pi_{{#1}}({#2})}

\begin{document}


\date{}

  \makeatletter
  \let\@copyrightspace\relax
  \makeatother

\title{Provenance Views for Module Privacy}

%

\author{Susan B. Davidson\thanks{University of Pennsylvania; Philadelphia, PA, USA. {\tt susan@cis.upenn.edu}}
~~~~~~Sanjeev Khanna\thanks{University of Pennsylvania; Philadelphia, PA, USA. {\tt sanjeev@cis.upenn.edu}}
~~~~~~Tova Milo\thanks{Tel Aviv University; Tel Aviv, Israel. {\tt milo@cs.tau.ac.il}}
~~~~~~Debmalya Panigrahi\thanks{CSAIL, MIT; Cambridge, MA, USA. {\tt debmalya@mit.edu}}
~~~~~~Sudeepa Roy\thanks{University of Pennsylvania; Philadelphia, PA, USA. {\tt sudeepa@cis.upenn.edu}}}


\maketitle


\begin{abstract}
\vspace{-1mm}
Scientific workflow systems increasingly store provenance
information about the module executions used to produce a data item,
as well as the parameter settings and intermediate data items passed
between module executions. However, authors/owners of workflows may
wish to keep some of this information confidential. In particular, a
{\em module} may be proprietary, and users should not be able to
infer its behavior by seeing mappings between all data inputs and
outputs.

The problem we address in this paper is the following:  Given a
workflow, abstractly modeled by a relation $R$, a privacy
requirement $\Gamma$ and costs associated with data. The
\emph{owner} of the workflow decides which data (attributes) to
hide, and provides the {\em user} with a view $R'$ which is the
projection of $R$ over attributes which have {\em not} been hidden.
{\em The goal is to minimize the cost of hidden data while
guaranteeing that individual modules are $\Gamma$-private.} We call
this the \secureview\ problem.  We formally define the problem,
study its complexity, and offer algorithmic solutions.

\end{abstract}

\vspace{-1mm}

\section{Introduction}
\label{sec:intro} \vspace{-1mm} The importance of data provenance
has been widely recognized. In the context of scientific workflows,
systems such as myGrid/Taverna  \cite{myGridWF}, Kepler
\cite{absModelKepler}, and VisTrails \cite{VisTrails} now capture
and store provenance information, and a standard for provenance
representation called the Open Provenance Model (OPM)
\cite{MoreauFFMMP08} has been designed. By maintaining information
about the module executions (processing steps) used to produce a
data item, as well as the parameter settings and intermediate data
items passed between module executions, the validity and reliability
of data can be better understood and results be made reproducible.

However, authors/owners of workflows may wish to keep some of this
provenance information private. For example, intermediate {\em data}
within an execution may contain sensitive information, such as the
social security number, a medical record, or financial information
about an individual.  Although users with the appropriate level of
access may be allowed to see such confidential data, making it
available to all users is an unacceptable breach of privacy.  Beyond
data privacy, a {\em module} itself may be proprietary, and hiding
its description may not be enough: users without the appropriate
level of access should not be able to infer its functionality by
observing all inputs and outputs of the module. Finally,  details of
how certain modules in the workflow are connected may be
proprietary, and therefore showing how data is passed between
modules may reveal too much of the {\em structure} of the workflow.
{\em There is thus an inherent trade-off between the utility of
provenance information and the privacy guarantees that
authors/owners desire.}

While data privacy was studied in the context of statistical
databases and ideas related to structural privacy were dealt with in
the context of workflow views, module privacy has not been addressed
yet. Given the importance of the issue \cite{cidr11}, this paper
therefore focuses on the problem of preserving the privacy of {\em
module functionality}, i.e. the mapping between input and output
values produced by the module (rather than the actual
\emph{algorithm} that \mbox{implements it).}

Abstracting the workflow models in
\cite{myGridWF,absModelKepler,VisTrails}, we consider a module to be
a finite relation which takes a set $I$ of input data (attributes),
produces a set $O$ of output data (attributes), and satisfies the
functional dependency $I \longrightarrow O$. A row in this relation
represents one execution of the module. In a {\em workflow},  $n$
such data processing modules are connected in a directed acyclic
multigraph (network), and jointly produce a set of final outputs
from a set of initial inputs. Each module receives input data from
one or more modules, or from the initial external input, and sends
output data to one or more modules, or produces the final output.
Thus a workflow can be thought of as a relation which is the
input-output join of the constituent module relations.  Each row in
this relation represents a workflow execution, and captures the
provenance of data that is produced during that execution. We call
this the {\em provenance relation.}.

To ensure the privacy of module functionality,  we extend the notion
of $\ell$-diversity~\cite{MKG+07} to our network setting\footnote{In
the Related Work, we discuss why a stronger notion of privacy, like
differential privacy, is not suitable
 here.}:  A module
with functionality  $m$ in a workflow is said to be $\Gamma$-private
if for every input $x$, the actual value of the output $m(x)$ is
indistinguishable from $\Gamma-1$ other possible values w.r.t. the
visible data values in the provenance relation. This is achieved by
carefully selecting a subset of data items and hiding those values
in \emph{all} executions of the workflow -- i.e. by showing the user
a {\em view} of the provenance relation for the workflow in which
the selected data items (attributes) are hidden. $\Gamma$-privacy of
a module ensures that even with arbitrary computational power and
access to the view for all possible executions of workflow, an
adversary can not guess the correct value of $m(x)$ with probability
$> \frac{1}{\Gamma}$.

Identical privacy guarantees can be achieved by hiding different
subsets of data. To reflect the fact that some data may be more
valuable to the user than other data, we assign a {\em cost} to each
data item in the workflow, which indicates the utility lost to the
user when the data value is hidden. It is important to note that,
due to {\em data sharing} (i.e. computed data items that are passed
as input to more than one module in the workflow), hiding some
data can be used to guarantee privacy for more than one module in
the network.

The problem we address in this paper is the following:  We are given
a workflow, abstractly modeled by a relation $R$, a privacy
requirement $\Gamma$ and costs associated with data.
An instance of $R$ represents the set of workflow executions that have been run.
The
\emph{owner} of the workflow decides which attributes to hide, and
provides the {\em user} with a view $R'$ which is the projection of
$R$ over the visible attributes.  {\em The goal is to minimize the
cost of hidden data while guaranteeing that individual modules are
$\Gamma$-private.} \eat{
 \footnote{For simplicity of
presentation, we assume a uniform $\Gamma$ value for all private
modules. However, all algorithms in the paper also work when
$\Gamma$  can be different for different modules.}
}
 We call this the secure-view 
problem.  We formally define the problem, study its
complexity, and offer algorithmic solutions.

\medskip
\noindent {\bf Contributions.~~} Our first contribution is to
formalize the notion of $\Gamma$-privacy of a private module  when
it is a standalone entity (\emph{standalone privacy}) as well as
when it is a component of a workflow interacting with other 
modules (\emph{workflow privacy}). For standalone
modules, we then analyze the computational and  communication
complexity of obtaining a minimal cost set of input/output data
items to hide such that the remaining, visible attributes  guarantee
$\Gamma$-privacy  (a {\em safe subset}).   We call this the
\emph{standalone secure-view} problem.

Our second set of contributions is to study workflows  in which all
modules are {\em private}, i.e.  modules for which the user has no a priori knowledge
and whose behavior must be hidden. For such
\emph{all-private workflows}, we analyze the
complexity of finding a minimum cost set of data items in the
workflow, as a whole, to hide such that the remaining visible
attributes guarantee $\Gamma$-privacy  for all modules. We call this
the {\em workflow secure-view} problem.  Although  the privacy of a
module within a workflow is inherently linked to the workflow
topology and functionality of other modules, we are able to show
that guaranteeing workflow secure-views in  this setting essentially
reduces to implementing the standalone privacy requirements for each
module.   We then study two variants of the workflow secure-view
problem, one in which module privacy is specified in terms of
attribute sets ({\em set constraints}) and one in which module
privacy is specified in terms of input/output cardinalities ({\em
cardinality constraints}). Both variants are easily shown to be
NP-hard,  and we give poly-time approximation algorithms
 for these problems. While the cardinality constraints version
 has an linear-programming-based $O(\log n)$-approximation algorithm, the set constraints version is much harder to approximate.
However, both variants becomes
more tractable when the workflow has \emph{bounded data sharing},
i.e. when a data item acts as input to a small number of modules.
In this case a constant factor approximation
is possible, although the problem remains NP-hard even without
any data sharing

Our third set of contributions is in \emph{general workflows}, i.e
workflows which contain private modules as well as modules whose
behavior is known ({\em public} modules).  Here we show that
ensuring standalone privacy of private modules no longer guarantees
their workflow privacy.  However, by  making some of the public
modules private ({\em privatization}) we can attain workflow privacy
of all private modules in the workflow. Since  privatization has a
cost,  the optimization problem, becomes much harder:  Even without
data sharing the problem is $\Omega(\log n)$-hard to approximate.
However, for both all-private and general workflows, there is
an LP-based $\ell_{\max}$-approximation algorithm, where
$\ell_{\max}$ is the length of longest requirement list for any module.

\medskip
\noindent {\bf Related Work.~~} {\em Workflow privacy} has been
considered in \cite{CCL+08, GF10, CG07}.  In \cite{CCL+08}, the
authors discuss a framework to output a {\em partial} view of a
workflow that conforms to a given set of access permissions on the
connections between modules and data on input/output ports. The
problem of ensuring the {\em lawful use} of data according to
specified privacy policies has been considered in \cite{GF10, CG07}.
The focus of the work is a policy language for specifying
relationships among data and module sets, and their properties
relevant to privacy. Although all these papers address workflow
privacy, the privacy notions are somewhat informal and no guarantees
on the quality of the solution are provided in terms of privacy and
utility. Furthermore, our work is the first, to our knowledge, to
address module privacy rather than data privacy.
\par
{\em Secure provenance} for workflows has been  studied
in~\cite{LM10, BSS08, HSW07}. The goal is to ensure that provenance
information has not been forged or corrupted, and a variety of
cryptographic and trusted computing techniques are proposed. In
contrast, we assume that provenance information has not been
corrupted, and focus on ensuring module privacy.
\par
In \cite{MS07}, the authors study information  disclosure in data
exchange, where given a set of public views, the goal is to decide
if they reveal any information about a private view. This does not
directly apply to our problem, where the private elements are the
$(\tup{x},m(\tup{x}))$ relations. For example, if all $\tup{x}$
values are shown without showing any of the $m(\tup{x})$ values for
a module $m$, then information is revealed in their setting but 
not in our setting.\footnote{In contrast, it can be shown that
showing all $m(\tup{x})$ values while hiding the $\tup{x}$'s, may
reveal information in our setting.} \eat{ Also our application does
not relate to the well-studied \emph{secure multi-party computation
of functions} \cite{multiparty} where a group of agents want to
jointly compute a function without revealing the information they
individually own. }
\par
{\em Privacy-preserving data mining} has received considerable
attention (see surveys~\cite{DataM, VBF+04}).
The goal is to hide individual data attributes while
retaining the suitability of data for mining patterns. For example,
the technique of \emph{anonymizing data} makes each record
indistinguishable from a large enough set of other records in
certain identifying attributes~\cite{Swee02, MKG+07, GFK+06}.
Privacy preserving approaches were studied for \emph{social
networks}~\cite{BDK07, RHM+09} 
\emph{auditing
queries} \cite{MNT08} 
and in other contexts.
Our notion of {\em standalone} module privacy is close to that of
$\ell$-diversity~ \cite{MKG+07}, in which the values of
\emph{non-sensitive attributes} are generalized so that, for every
such generalization, there are at least $\ell$ different values of
\emph{sensitive attributes}.   We extend this work in two ways:
First, we place modules (relations) in a network of modules, which
significantly complicates the problem,  Second, we analyze the
complexity of attaining standalone as well as workflow \mbox{
privacy of modules.} \eat{ However, to the best of our knowledge,
the complexity analysis of attaining $\ell$-diversity has not been
studied in the literature. More importantly,  we extend this notion
to workflows where attaining $\ell$-diversity becomes non-trivial
due to interaction with other modules. }
\par
Another widely used technique is that of \emph{data perturbation}
where some noise (usually random) is added to the the output of a
query or to the underlying database. This technique is often used in
\emph{statistical databases}, where a query computes some aggregate
function over the dataset~\cite{DN03} 
and the goal is to
preserve the privacy of data elements. \eat{ These techniques seem
unlikely to extend to our setting where  users query the value of
data items in the workflow whereas the private elements are
$(\tup{x}, m(\tup{x}))$ pairs for a private module $m$ and input
assignments $\tup{x}$ to $m$, which leads to direct conflict between
the utility measure to the user and privacy measure to the workflow
owner. } In contrast, in our setting the private elements are
$(\tup{x}, m(\tup{x}))$ pairs for a private module $m$ and the queries
are select-project-join style queries over the provenance relation
rather than aggregate queries.

\par
Privacy in {\em statistical databases} is typically quantified
using {\em differential privacy}, which requires that the output
distribution is {\em almost} invariant to the inclusion of any
particular record 
(see surveys~\cite{Dwork08, Dwork09} and the references therein). Although this is the strongest
notion of privacy known to date,
 \emph{no} deterministic algorithm can guarantee differential privacy.
Thus differential privacy is unsuitable for our purposes, since
adding random noise to provenance information may render it useless;
provenance is used to ensure reproducibility of experiments  and
therefore data values must be accurate.  Our approach of outputting
a safe view allows the user to know the name of all data items and
the exact values of data that is visible.
The user also does not lose any utility in terms of 
\emph{connections} in the workflow, and can infer exactly which
module produced which visible data item or whether two visible data
items depend on each other. \eat{Only information the user is unable
to see is the values of of some of the data items which are
consistently hidden in all executions of the workflow. }

\medskip
\noindent {\bf Organization.} Section~\ref{sec:model} defines our
workflow model and formalizes the notions of $\Gamma$-privacy of a
module, both when it is standalone and when it appears in a
workflow. The secure-view problem for standalone module privacy is
studied in Section~\ref{sec:standalone}.
Section~\ref{sec:private-module} then studies the problem for
workflows consisting only of private modules, whereas
Section~\ref{sec:public-module} generalizes the results to general
workflows consisting of both public and private modules. Finally we
conclude and discuss directions for future work in
Section~\ref{sec:conclusions}.

\section{Preliminaries}\label{sec:model}
We start by introducing some notation and
formalizing our notion of privacy. We first consider the privacy of a
single module, which we call {\em standalone module privacy}. Then
we consider privacy when modules are connected
in a workflow, which we call \emph{workflow module privacy}.

\subsection{Modules and Relations}

We model a module $m$ with a set $I$ of input variables and a set
$O$ of (computed) output variables as a relation $R$
over a set of attributes $A=I \cup O$ that satisfies the functional
dependency $I \rightarrow O$.  In other words, $I$ serves as a (not
necessarily minimal) key for $R$.  We assume that $I \cap O =
\emptyset$ and will refer to $I$ as the \emph{input attributes} of
$R$ and to $O$ as its \emph{output attributes}.

\par

We assume that the values of each attribute $a \in A$ come from a
finite but arbitrarily large domain $\Delta_{a}$, and let $\Dom =
\prod_{a \in I}\Delta_{a}$ and $\Range = \prod_{a \in O}\Delta_{a}$
denote the {\em domain} and {\em range} of the module $m$
respectively. The relation $R$ thus represents the (possibly
partial) function $m: \Dom \rightarrow \Range$ and tuples in $R$
describe executions of $m$, namely for every $t \in R$,
$\proj{O}{\tup{t}}=m(\proj{I}{\tup{t}})$.
We overload the standard
notation for projection, $\proj{A}{R}$, and use it for a tuple
$\tup{t} \in R$. Thus $\proj{A}{\tup{t}}$, for a set $A$ of
attributes, denotes the projection of $\tup{t}$ to the attributes in
$A$.


\begin{example}\label{ex:module}
Figure~\ref{fig:wf-view} shows a simple workflow involving three
modules $m_1,m_2,m_3$ with boolean input and output attributes; we
will return to it shortly and focus for now on the top module $m_1$.
Module $m_1$ takes as input two data items, $a_1$ and $a_2$, and
computes $a_3\! =\! a_1\!\vee\! a_2$, $a_4\! =\! \neg({a_1\!\wedge\!
a_2})$ and $a_5\! =\! \neg({a_1\!\oplus \!a_2})$. (The symbol
$\oplus$ denotes XOR). The relational representation (functionality)
of module $m_1$ is shown in Figure~\ref{fig:m1} as relation $R_1$,
with the functional dependency $a_1 a_2 \longrightarrow a_3 a_4
a_5$. For clarity, we have added $I$ (input) and $O$ (output) above
the attribute names to indicate their role.
\end{example}



%
%
%

\begin{figure}[ht!]
\begin{center}
\subfloat
{\label{fig:wf1}
\begin{tabular}{c}
    \includegraphics[scale=.4]{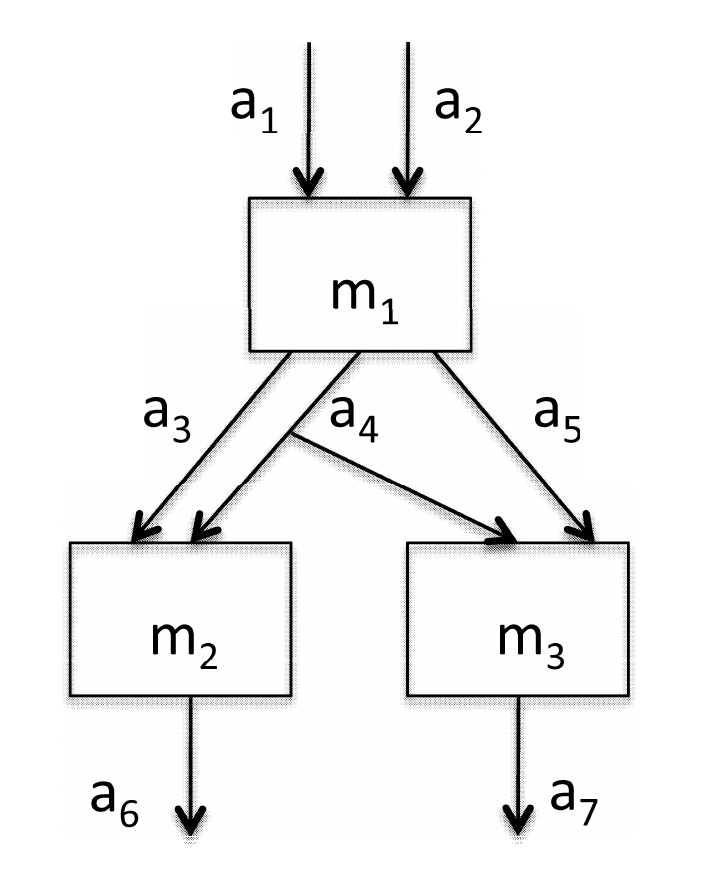}
\end{tabular}
  }
~~~
 \subfloat [{\small $R$:  Workflow executions}  \vspace{-5.3mm}]{
\begin{tabular}{|lllllll|} \hline
$a_1$ & $a_2$ & $a_3$  & $a_4$ & $a_5$ & $a_6$& $a_7$\\  \hline
\hline 0 & 0 & 0 & 1 & 1 & 1& 0\\ \hline 0 & 1 & 1 & 1 & 0 & 0 & 1\\
\hline 1 & 0 & 1 & 1 & 0 & 0 &1\\ \hline 1 & 1 & 1 & 0 &1 & 1 &1
\\\hline
\end{tabular}\label{fig:wf-prov} }
~~~
\subfloat[{\small $R_1$: Functionality of $m_1$}]{\label{fig:m1}
\begin{tabular}{|cc|ccc|}
\hline \multicolumn{2}{|c|}{$I$} & \multicolumn{3}{c|}{$O$}\\ \hline
$a_1$ & $a_2$ & $a_3$ & $a_4$ & $a_5$ \\ \hline \hline 0 & 0 & 0 & 1
& 1 \\ \hline 0 & 1 & 1 & 1 & 0 \\ \hline 1 & 0 & 1 & 1 & 0 \\
\hline 1 & 1 & 1 & 0 & 1 \\ \hline
\end{tabular}
}
 ~~~ 
 \subfloat[{\small~ \!\!$R_V \!=\!\proj{V}{R_1}$
 }]{\label{fig:m1'}
\begin{tabular}{|c|cc|}
\hline \multicolumn{1}{|c|}{$I \cap V$} & \multicolumn{2}{c|}{$O
\cap V$}\\ \hline $a_1$ & $a_3$ & $a_5$ \\ \hline  \hline 0 & 0 & 1
\\ \hline 0 & 1 & 0 \\ \hline 1 & 1 & 0 \\ \hline 1 & 1 & 1 \\
\hline
\end{tabular}
} 
 \vspace{-4mm} \caption{A sample workflow, and, workflow and module executions as
relations} \vspace{-5mm} \label{fig:wf-view}
\end{center}
\end{figure}

\subsection{Standalone Module Privacy}

Our approach to ensuring standalone module privacy, for a module
represented by the relation $R$, will be to hide a carefully chosen
subset of $R$'s attributes. In other words, we will project $R$ on a
restricted subset $V$ of attributes (called the {\em visible
attributes} of $R$), allowing users access only to the view
$R_V=\proj{V}{R}$. The remaining, non visible, attributes of $R$ are
called {\em hidden attributes}.

\par
We distinguish below two types of modules. (1) {\em Public modules}
whose behavior is fully known to users. Here users have a priori
knowledge about the full  content of $R$ and, even if given only
the view $R_V$, they are able to fully (and exactly) reconstruct
$R$. Examples include reformatting or sorting modules.
  (2) Private modules where such a priori knowledge does not
exist. Here, the only information available to users, on the
module's behavior, is the one given by $R_V$. Examples include
proprietary software,  e.g. a genetic disorder susceptibility
module\footnote{We discuss in Section \ref{sec:conclusions} how
partial prior knowledge can be handled by our approach.}.

Given a view (projected relation) $R_V$ of a private module $m$, the
\emph{possible worlds} of $m$ are all the possible full relations
(over the same schema as $R$) that are consistent with $R_V$ w.r.t
the visible attributes. Formally,
\begin{definition}\label{def:pos-worlds-standalone}
Let $m$ be a private module with a corresponding relation $R$,
having input and output attributes $I$ and $O$, resp., and let $V
\subseteq I \cup O$. The set of {\em possible worlds} for $R$ w.r.t. $V$,
denoted $\Worlds(R, V)$, consist of all relations
$R'$ over the same schema as $R$ that satisfy the functional
dependency $I \rightarrow O$ and where $\proj{V}{R'} = \proj{V}{R}$. 
\end{definition}
\begin{example}\label{eg:sa-worlds}
Returning to module $m_1$, suppose the visible attributes are $V =
\{a_1, a_3, a_5\}$ resulting in the view $R_V$ in
Figure~\ref{fig:m1'}. For clarity, we have added $I \cap V$ (visible
input) and $O \cap V$ (visible output) above the attribute names to
indicate their role. Naturally, $R_1 \in \Worlds(R_1, V)$.
Figure~\ref{fig:possible-worlds} shows four additional sample
relations $R_1^1,R_1^2,R_1^3, R_1^4$ in $\Worlds(R_1, V)$, such that $\forall i
\in [1, 4], \proj{V}{R_1^i} = \proj{V}{R_1}= R_V$. (Overall there are
sixty four relations in $\Worlds(R_1, V)$).
\end{example}
To guarantee privacy of a module $m$, the view $R_V$ should ensure
some level of uncertainly w.r.t the
 value of the output
$m(\proj{I}{\tup{t}})$, for tuples $t\in R$. To define this, we
introduce the notion of $\Gamma$-standalone-privacy, for a given
parameter $\Gamma \geq 1$. Informally, $R_V$ is
$\Gamma$-standalone-private if for every $t \in R$, the possible
worlds $\Worlds(R, V)$ contain at least $\Gamma$ distinct output
values that could be the result of $m(\proj{I}{\tup{t}})$.
\begin{definition}\label{def:standalone-privacy}
Let $m$ be a private module with a corresponding relation $R$ having
input and output attributes $I$ and $O$ resp. Then $m$ is
\emph{$\Gamma$-standalone-private} w.r.t a set of visible attributes
$V$, if for every tuple $x \in \proj{I}{R}$, \mbox{ $|\Out_{x, m}|
\geq \Gamma$}, where
  $\Out_{x, m} = \{\tup{y} \mid  \exists R' \in \Worlds(R, V),~~\exists \tup{t'} \in R'~~ s.t~~
  \tup{x} = \proj{I}{\tup{t'}} \wedge \tup{y}=\proj{O}{\tup{t'}} \}$.
\par
If $m$ is \emph{$\Gamma$-standalone-private} w.r.t. $V$, then we will call $V$ a \emph{safe subset} for 
$m$ and $\Gamma$.
\end{definition}

Practically,
$\Gamma$-standalone-privacy means that for any input
the adversary cannot guess $m$'s output with probability greater
than $\frac{1}{\Gamma}$.

\begin{example}\label{eg:sa-privacy}
It can be verified that, if $V = \{a_1, a_3, a_5\}$ then for all $\tup{x} \in \proj{I}{R_1}$,
$|Out_x| \geq 4$, so
$\{a_1, a_3, a_5\}$ is safe for $m_1$ and $\Gamma = 4$.
As an example, from
Figure~\ref{fig:possible-worlds}, when $\tup{x} = (0, 0)$,
 $\Out_{x,m} = \{ (0, \underline{0}, 1), (0, \underline{1}, 1), (1, \underline{0}, 0), (1, \underline{1}, 0) \}$
 (hidden attributes are underlined).
Also, hiding any two output attributes from $O = \{a_3, a_4, a_5\}$
ensures standalone privacy for $\Gamma = 4$. For example, if $V =
\{a_1, a_2, a_3\}$ (i.e. the output attributes $\{a_4, a_5\}$ are
hidden), then the input $(0, 0)$ can be mapped to one of $(0,
\underline{0}, \underline{0}), (0, \underline{0}, \underline{1}),
(0, \underline{1}, \underline{0})$ and $(0, \underline{1},
\underline{1})$; this holds for other assignments of input
attributes as well. However, $V = \{a_3, a_4, a_5\}$ (i.e. when only
the input attributes are hidden) is not safe for $\Gamma = 4$: for
any input $\tup{x}$, $\Out_{x, m} = \{(0,1,1), (1,1,0), (1,0,1)\}$,
containing only three possible output tuples.
\end{example}

There may be several safe subsets $V$ for a given module $m$ and parameter $\Gamma$.
Some of the corresponding $R_V$ views
may be preferable to others, e.g. they provide users with
more useful information, allow to answer more common/critical user
queries, etc. Let $\widebar{V}=(I \cup O) \setminus V$ denote the
attributes of $R$ that do not belong to the view. If $\cost(\widebar{V})$
denotes the penalty of hiding the attributes in $\widebar{V}$,
a natural goal is to choose a view s.t. that $\cost(\widebar{V})$
is minimized.
To understand the difficulty of this problem, we study a
version of the problem where
the cost function is additive:
 each attribute $a$ has some penalty value $\cost(a)$
and the penalty of hiding $\widebar{V}$  is the sum of the penalties
of the individual attributes,  $\cost(\widebar{V}) = \Sigma_{a\in
\widebar{V}} \cost(a)$. We call this optimization problem the
\emph{standalone \secureview problem} and discuss it in
Section~\ref{sec:standalone}.



\subsection{Workflows and Relations}

A workflow $W$ consists of a set of modules $m_1, \cdots, m_n$,
connected as a DAG (see, for example, the workflow in
Figure~\ref{fig:wf-view}).  Each module $m_i$ has a set $I_i$ of
input attributes and a set $O_i$ of output attributes.
We assume that (1) for each module, the names of its
input and output attributes are disjoint, i.e. $I_i \cap O_i =
\emptyset$,  (2) the names of the output attributes of distinct
modules are disjoint, namely $O_i \cap O_j = \emptyset$, for $i\neq
j$ (since each data item is produced by a unique module), and
(3) whenever an output of a module $m_i$ is fed as input to
a module $m_j$ the corresponding output and input attributes of
$m_i$ and $m_j$ have the same name. The DAG shape of the workflow
guarantees that these requirements are not contradictory.

\eat{
To execute the workflow, some input values need to be provided to
the input variables which are not fed by other modules (e.g. attributes $a_1$ and $a_2$ in
Figure~\ref{fig:wf1}).  Modules
where all input variables were assigned some value can be executed.
The computed output values are then fed to the relevant connected
modules. This enables the execution of further modules where all
input variables now have some assigned value, etc.
}

We model executions of $W$ as a relation $R$ over the set of
attributes $A=\cup_{i=1}^{n} (I_i \cup O_i)$, satisfying the set
of functional dependencies $F = \{I_i \rightarrow O_i : i \in [1, n]\}$.
Each tuple in $R$ describes an execution of the workflow $W$. In
particular, for every $t \in R$, and every $i \in [1, n]$,
$\proj{O_i}{\tup{t}}=m_i(\proj{I_i}{\tup{t}})$.

\begin{example}\label{eg:wf-relations}
Returning to the sample workflow in Figure~\ref{fig:wf-view}, the
input and output attributes of modules $m_1, m_2, m_3$ respectively
are (i) $I_1 = \{a_1, a_2\}$, $O_1 = \{a_3, a_4, a_5\}$, (ii) $I_2 =
\{a_3, a_4\}$, $O_2 = \{a_6\}$ and (iii) $I_3 = \{a_4, a_5\}$, $O_3
= \{a_7\}$. The underlying functional dependencies in the relation
$R$ in Figure~\ref{fig:wf-prov} reflect the keys of the constituent
modules, e.g. from $m_1$ we have $a_1 a_2 \longrightarrow a_3 a_4
a_5$, from $m_2$ we have $a_3 a_4 \longrightarrow a_6$, and from
$m_3$ we have $a_4 a_5 \longrightarrow a_7$.
\end{example}

Note that the output of a module may be
input to several modules, hence the names of the input attributes of
distinct modules are not necessarily disjoint. It  is therefore possible
that $I_i\cap I_j \neq \emptyset$ for $i \neq j$. We call this {\em
data sharing} and  define the degree of data sharing
in a workflow:

\begin{definition}
A workflow $W$ is said to have \emph{$\gamma$-bounded data sharing}
if every attribute in $W$ can appear in the left hand side of at
most $\gamma$ functional dependencies $I_i \rightarrow O_i$.
\eat{Equivalently, if $I_i \rightarrow O_i$, $i = 1$ to $n$ are the
functional dependencies in relation $R$ of the workflow, every
attribute $a$ of $R$ can appear on the left hand side of at most
$\gamma$ such dependencies.}
\end{definition}

In the workflow of our running example, $\gamma = 2$. Intuitively,
if a workflow has $\gamma$-bounded data sharing then a data item can
be fed as input to at most $\gamma$ different modules. In the
following sections we will see the implication of such a bound on
the complexity of the problems studied.

\begin{figure}[h]
\centering \subfloat[$R_1^1$]{\label{tab:w1}
\begin{tabular}{|cc|ccc|}
\hline \multicolumn{2}{|c|}{$I$} & \multicolumn{3}{c|}{$O$}\\ \hline
$a_1$ & $a_2$ & $a_3$ & $a_4$ & $a_5$ \\ \hline 0 & 0 & 0 & 0 & 1 \\
\hline 0 & 1 & 1 & 0 & 0 \\ \hline 1 & 0 & 1 & 0 & 0 \\ \hline 1 & 1
& 1 & 0 & 1 \\ \hline
\end{tabular}
} \subfloat[$R_1^2$]{\label{tab:w2}
\begin{tabular}{|cc|ccc|}
\hline \multicolumn{2}{|c|}{$I$} & \multicolumn{3}{c|}{$O$}\\ \hline
$a_1$ & $a_2$ & $a_3$ & $a_4$ & $a_5$ \\ \hline 0 & 0 & 0 & 1 & 1 \\
\hline 0 & 1 & 1 & 1 & 0 \\ \hline 1 & 0 & 1 & 0 & 0 \\ \hline 1 & 1
& 1 & 0 & 1 \\ \hline
\end{tabular}
} 
\subfloat[$R_1^3$]{\label{tab:w3}
\begin{tabular}{|cc|ccc|}
\hline \multicolumn{2}{|c|}{$I$} & \multicolumn{3}{c|}{$O$}\\ \hline
$a_1$ & $a_2$ & $a_3$ & $a_4$ & $a_5$ \\ \hline 0 & 0 & 1 & 0 & 0 \\
\hline 0 & 1 & 0 & 0 & 1 \\ \hline 1 & 0 & 1 & 0 & 0 \\ \hline 1 & 1
& 1 & 0 & 1 \\ \hline
\end{tabular}
}
\subfloat[$R_1^4$]{\label{tab:w4}
\begin{tabular}{|cc|ccc|}
\hline \multicolumn{2}{|c|}{$I$} & \multicolumn{3}{c|}{$O$}\\ \hline
$a_1$ & $a_2$ & $a_3$ & $a_4$ & $a_5$ \\ \hline 0 & 0 & 1 & 1 & 0 \\
\hline 0 & 1 & 0 & 1 & 1 \\ \hline 1 & 0 & 1 & 0 & 0 \\ \hline 1 & 1
& 1 & 0 & 1 \\ \hline
\end{tabular}
} \vspace{-2mm} \caption{$R_1^i \in \Worlds(R_1, V)$, $i \in [1,
4]$}
\vspace{-4mm} \label{fig:possible-worlds}
\end{figure}

\eat{
\scream{TODO: (i) Add example of workflow GRAPH, and details the name of
input/output variables of each module.~~~\\
(ii) Show an example for an execution (explain where each
value comes from).\\
(iii) Show the relation and FDs for the workflow DAG in the
previous example, and the tuple describing its execution.\\
(iv) Mention somewhere that all the results also hold when $\Gamma_i$
is different for different modules $m_i$}
}

\subsection{Workflow Module Privacy}
\eat{
}

To define privacy in the context of a workflow, we first extend our
notion of {\em possible worlds} to a workflow view. Consider the
view $R_V=\proj{V}{R}$ of a workflow relation $R$. Since the
workflow may contain private as well as public modules, a possible
world for $R_V$ is a full relation that not only agrees with $R_V$
on the content of the visible attributes, but is also consistent
w.r.t the expected behavior of the public modules. In the following
definitions, $m_1, \cdots, m_n$ denote the modules in the workflow
$W$ and $F$ denotes the set of functional dependencies $I_i\!
\rightarrow \!O_i$, $i \in [1, n]$ in the relation $R$.

\begin{definition}\label{def:pos-worlds-workflow}
\eat{
} The set of {\em possible worlds} for the workflow relation $R$
w.r.t. $V$, denoted also $\Worlds(R, V)$, consists of all the
relations $R'$ over the same attributes as $R$ that satisfy the
functional dependencies in $F$
and where (1) $\proj{V}{R'} = \proj{V}{R}$, 
and (2) for every public module
$m_i$ in $W$ and every tuple $\tup{t'} \in R'$,
$\proj{O_i}{\tup{t'}} = m_i(\proj{I_i}{\tup{t'}})$.
\end{definition}

Note that when a workflow consists only of private modules, the
second constraint does not need to be enforced. We call these
{\em all-private workflows} and study them in
Section~\ref{sec:private-module}. We then show in
Section~\ref{sec:public-module}  that attaining privacy when public
modules are also used is fundamentally harder.

We are now ready to define the notion of  $\Gamma$-workflow-privacy,
for a given parameter $\Gamma \geq 1$. Informally, a view $R_V$ is
$\Gamma$-workflow-private if for every tuple $t \in R$, and every
private module $m_i$ in the workflow, the possible worlds
$\Worlds(R, V)$ contain at least $\Gamma$ distinct output values that
could be the result of $m_i(\proj{I_i}{\tup{t}})$.

\begin{definition}\label{def:workflow-privacy}
A private module $m_i$ in $W$ is \emph{$\Gamma$-workflow-private}
w.r.t a set of visible
attributes $V$, if for every tuple \mbox{$\tup{x} \in \proj{I_i}{R}$},
  \mbox{$|\Out_{\tup{x}, W}| \geq \Gamma$}, $\text{~~where~~}$ 
  $\Out_{\tup{x}, W} =$ $\{\tup{y} \mid  \exists R'$ $\in \Worlds(R, V),$ $s.t., ~~\forall \tup{t'} \in R'~~ 
  \tup{x} = \proj{I_i}{\tup{t'}}$ $\Rightarrow \tup{y}=\proj{O_i}{\tup{t'}} \}$.
  
$W$ is called \emph{$\Gamma$-private} if every private
module $m_i$ in $W$ is $\Gamma$-workflow-private.

If $W$ (resp. $m_i$) is \emph{$\Gamma$-private}
($\Gamma$-workflow-private) w.r.t. $V$, then we call $V$ a
\emph{safe subset} for $\Gamma$-privacy of $W$
($\Gamma$-workflow-privacy of $m_i$).
\end{definition}
For simplicity, in the above definition we assumed that
the privacy requirement of every module $m_i$ is the same $\Gamma$.
The results and proofs in this paper remain unchanged
when different modules $m_i$ have different privacy requirements
$\Gamma_i$.

In the rest of the paper, for a set of
visible attributes $V \subseteq A$, $\overline{V} = A \setminus V$
will denote the hidden attributes in the workflow.
\eat{
}
The following proposition is easy to verify, which will be useful
later:
\eat{
 which says that if a set
of hidden
attributes guarantees $\Gamma$-workflow-privacy
for a module $m_i$,
then hiding any of its superset also guarantees $\Gamma$-workflow-privacy for $m_i$.
}

\begin{proposition}\label{prop:superset}
If $V$ is a safe subset for $\Gamma$-workflow-privacy of a module $m_i$
in $W$, then any $V'$ such that $V' \subseteq V$ (or, $\widebar{V'} \supseteq \widebar{V}$) also
guarantees $\Gamma$-workflow-privacy of $m_i$.
\end{proposition}

As we illustrate in the sequel, given a workflow $W$ and a parameter
$\Gamma$ there may be several incomparable (in terms of set
inclusion) sets $V$ of visible attributes w.r.t. which $W$ is
$\Gamma$-private. Our goal will be to choose one that minimizes the
penalty $\cost(\widebar{V}) = \Sigma_{a\in \widebar{V}} \cost(a)$ of
the  hidden attributes $\widebar{V}$ -- this we call the \emph{workflow \secureview problem},
or simply the \secureview problem. The candidates
are naturally the maximal, in terms of set inclusion, safe sets
$V$ (and correspondingly the minimal $\widebar{V}$s).


\subsection{Complexity Classes and Approximation}
\noindent
In the following sections
we will study the \secureview problem: minimize
cost of the hidden attributes that ensures that a workflow
is $\Gamma$-private.
We will prove that this problem is
NP-hard even in very restricted cases and study poly-time
\emph{approximation algorithms} as well as the \emph{hardness of
approximations} for different versions of the problem. We will use
the following common notions of complexity and approximation:
\eat{, common in the
literature (for instance, see \cite{VaziraniBook}).}
An algorithm is said to be a
\emph{$\mu(n)$-approximation algorithm} for a given optimization
problem, for some non-decreasing function $\mu(n): \mathbb{N}^+
\rightarrow \mathbb{N}$, if on every input of size $n$ it computes a
solution where the value is within a factor of $\mu(n)$ of the value
returned by an optimal algorithm for the problem. An optimization
problem is said to be \emph{$\mu(n)$-hard to approximate} if for all
inputs of size $n$, for all sufficiently large $n$, a
$\mu(n)$-approximation algorithm for the problem cannot exist
assuming some standard complexity results hold. In some parts of
the paper (Theorems~\ref{thm:cardinality-private},
\ref{thm:set-private}, \ref{thm:cardinality-public}), we will use
complexity results of the form $NP \not\subseteq DTIME(n^{f(n)})$,
where $f(n)$ is a poly-logarithmic or sub-logarithmic function of $n$
and $DTIME$ represents deterministic time. For example, the hardness
result in Theorem~\ref{thm:cardinality-private} says that there
cannot be an $O(\log n)$-approximation algorithm unless all problems
in $NP$ have $O(n^{\log \log n})$-time deterministic exact
algorithms. Finally, a problem is said to be \emph{APX-hard} if
there exists a constant $\epsilon > 0$ such that a $(1 +
\epsilon)$-approximation in poly-time would imply $P = NP$. If a
problem is APX-hard, then the problem cannot have a \emph{PTAS},
i.e, a $(1+\epsilon)$-approximation algorithm which runs in
polynomial time for all constant $\epsilon > 0$, unless $P = NP$.


\section{Standalone Module Privacy}\label{sec:standalone}
We start our study of workflow privacy by considering the privacy of
a standalone module, which is the  simplest special case of a
workflow. Hence understanding it is a first step towards
understanding the general case. We will also see that
standalone-privacy guarantees of individual modules may be used as
building blocks for attaining workflow privacy.

We analyze below the time complexity of obtaining (minimal cost)
guarantees for standalone module privacy. Though the notion of
$\Gamma$-standalone-privacy is similar to the well-known notion of
$\ell$-diversity \cite{MKG+07}, to the best of our knowledge the
time complexity of this problem has not been studied.\\

\noindent
\textbf{Optimization problems and parameters.~~}
 Consider a standalone module $m$ with input
attributes $I$, output attributes $O$, and a relation $R$. Recall
that a visible subset of attributes $V$ is called a \emph{safe
subset} for module $m$ and privacy requirement $\Gamma$, if $m$ is
$\Gamma$-standalone-private w.r.t. $V$ (see
Definition~\ref{def:standalone-privacy}). If each attribute $a \in I
\cup O$ has cost $\cost(a)$, \emph{the standalone \SAsecureview\ problem}
aims to find a safe subset $V$ s.t. the cost of the hidden
attributes, $\cost(\widebar{V}) = \sum_{a \in \widebar{V}}
\cost(a)$, is minimized. The corresponding decision version will
take a cost limit $C$ as an additional input, and decide whether
there exists a safe subset $V$ such that $\cost(\widebar{V}) \leq
C$.

One natural way of solving the optimization version of the
standalone \SAsecureview\ problem is to consider all possible
subsets $V \subseteq I \cup O$, check if $V$ is safe, and return the
safe subset $V$ s.t. $\cost(\widebar{V})$ is minimized. This
motivates us to define and study the simpler \emph{\safeview
problem}, which takes a subset $V$ as input and decides whether $V$
is a safe subset.

 To understand how much of the complexity of the standalone
\SAsecureview\ problem comes from the need to consider different
subsets of attributes, and what is due to the  need to determine the
safety of subsets, we study the time complexity of standalone \SAsecureview,
with and without access to \emph{an oracle for the \safeview
problem}, henceforth called a \emph{\safeview oracle}. A \safeview
oracle takes a subset $V \subseteq I \cup O$ as input and answers
whether $V$ is safe. In the presence of a \safeview oracle, the time
complexity of the \safeview problem is mainly due to the number of
oracle calls, and hence we study the \emph{communication
complexity}. Without access to such an oracle, we also study the
\emph{computational complexity} of this problem.

In our discussion below, $k = |I| + |O|$ denotes the total number of
attributes in the relation $R$, and $N$ denotes the number of rows
in $R$ (i.e. the number of executions). Then
$N \leq \prod_{a \in I}|\Delta_a|$ $\leq \delta^{|I|} \leq \delta^{k}$
where $\Delta_a$ is the domain of attribute $a$
and $\delta$ is the maximum domain size of attributes.

\subsection{Lower Bounds}\label{sec:lb_standalone}
\eat{
%
} We start with lower bounds for the \safeview problem. Observe that
this also gives lower bounds for the standalone \SAsecureview\
problem without a \safeview oracle. To see this, consider a set $V$
of attributes and assume that each attribute in $V$ has cost $>0$
whereas all other attributes have cost zero. Then \safeview has a
positive answer for $V$ iff the standalone \SAsecureview\ problem
has a solution with cost $=0$ (i.e. one that hides only the
attributes $\widebar{V}$).

\vspace{2mm}

 \noindent \textbf{Communication complexity of
\safeview.} Given a visible subset $V \subseteq I \cup O$, we show
that deciding whether $V$ is safe needs $\Omega(N)$ time.
Note that just to read the table as input takes $\Omega(N)$ time. So
the lower bound of $\Omega(N)$ does not make sense unless we assume
the presence of a \emph{data supplier} (we avoid using the term
``oracle'' to distinguish it from \safeview oracle) which supplies
the tuples of $R$ on demand:  Given an assignment $\tup{x}$ of the
input attributes $I$, the data supplier outputs the value $\tup{y} =
m(\tup{x})$ of the output attributes $O$. The following theorem
shows the $\Omega(N)$ communication complexity lower bound
in terms of the number of calls to the data
supplier; namely, that (up to a constant factor) one indeed needs to
view the full relation.
%
%

\begin{theorem}\label{thm:communication-no-oracle}
\textbf{(\safeview Communication Complexity) }
Given a module $m$, a subset $V \subseteq I \cup O$,
and a privacy requirement $\Gamma$, deciding
whether $V$ is safe for $m$ and $\Gamma$
requires $\Omega(N)$ calls to the data supplier, where
 $N$ is the number of tuples in the relation $R$ of $m$.
\end{theorem}

\begin{proof}[Proof sketch]
This theorem is proved by a reduction from the
\emph{set-disjointness problem}, where Alice and Bob hold two
subsets $A$ and $B$ of a universe $U$ and the goal is decide whether
$A \cap B \neq \phi$. This problem is known to have $\Omega(N)$
communication complexity where $N$ is the number of elements in the
universe. Details are  in
Appendix~\ref{sec:app_proofs_standalone}.
\end{proof}


 \noindent \textbf{Computational Complexity of \safeview: }
 The above $\Omega(N)$ computation complexity of  \safeview holds
when the relation $R$ is given explicitly tuple by tuple. The
following theorem shows that even when $R$
is described implicitly in a succinct manner, there
cannot be a poly-time (in the number of attributes) algorithm to
decide whether a given subset $V$ is safe unless $P = NP$
(proof is in Appendix~\ref{sec:app_proofs_standalone}).

\begin{theorem} \label{thm:computation-no-oracle}
\textbf{(\safeview Computational Complexity) } Given a module $m$
with a poly-size (in $k\!=\!|I|\!+\!|O|$) description  of
functionality, a subset $V \subseteq I \cup O$, and a privacy
requirement $\Gamma$, deciding whether $V$ is safe w.r.t. $m$ and
$\Gamma$ is co-NP-hard in $k$.
\end{theorem}

\begin{proof} [Proof sketch]
The proof of this theorem  works by a reduction from
the UNSAT problem, where given a boolean CNF formula $g$ on
variables $x_1, \cdots, x_{\ell}$, the goal is to decide whether, for
\emph{all} assignments of the variables, $g$ is not satisfiable.
Here given any assignment of the variables $x_1, \cdots, x_\ell$,
$g(x_1, \cdots, x_\ell)$ can be evaluated in polynomial time, which
simulates the function of the data supplier.
\end{proof}

\noindent
\textbf{Lower Bound of Standalone \SAsecureview\ with a \safeview Oracle:~~} 
Now suppose we have access to a \safeview oracle, which takes care
of the ``hardness'' of the \safeview problem given in
Theorems~\ref{thm:communication-no-oracle} and
\ref{thm:computation-no-oracle}, in constant time.
The oracle takes a visible subset $V \subseteq I \cup O$ 
as input, and answers whether $V$ is safe for module $m$ and privacy requirement $\Gamma$. 
 The following theorem shows that the decision version of standalone \SAsecureview\
 remains hard
(i.e. not solvable in poly-time in the number of attributes):

\begin{theorem}\label{thm:communication-oracle}
\textbf{(Standalone \SAsecureview\ Communication Complexity, with \safeview oracle) }
Given a \safeview oracle and a cost limit $C$,
deciding whether there exists a safe subset $V \subseteq I \cup O$ with cost bounded by  $C$
requires $2^{\Omega(k)}$ oracle calls, where $k = |I| + |O|$.
\end{theorem}

\begin{proof}[Proof sketch] 
The  proof of this theorem involves a novel construction of two
functions, $m_1$ and $m_2$, on $\ell$ input attributes and a single
output attribute, such that for $m_1$ the minimum cost of a safe
subset is $\frac{3\ell}{4}$ whereas for $m_2$ it is $\frac{\ell}{2}$
($C = \frac{\ell}{2}$). In particular, for both $m_1$ and $m_2$, all
subsets of size $< \frac{\ell}{4}$ are safe and all other subsets
are unsafe, except that for $m_2$, there is exactly one special
subset of size $\frac{\ell}{2}$ such that this subset and all
subsets thereof are safe.

%
\par
    We show that for an algorithm using $2^{o(k)}$ calls, there always
    remains at least one special subset of size $\frac{\ell}{2}$
    that is consistent
    with all previous answers to queries. Hence after $2^{o(k)}$
calls, if the algorithm decides that there is a safe subset with
cost $\leq C$, we choose $m$ to be $m_1$; on the other hand, if it
says that there is no such subset, we set $m = m_2$. In both the
cases the answer of the algorithm is wrong which shows that there
cannot be such an algorithm distinguishing these two cases with
$2^{o(k)}$ calls (details in
Appendix~\ref{sec:app_proofs_standalone}).
%
\end{proof}
\subsection{Upper Bounds}\label{sec:ub_standalone}
\eat{
We now present simple algorithms for solving the \SAsecureview\ and
\safeview problems, with time complexity polynomial in the lower
bounds given above.
}
The lower bound results given above show that solving the standalone
\secureview\ problem is unlikely in time sub-exponential in $k$
or sub-linear in $N$. We now present simple algorithms for
solving the \SAsecureview\ and
\safeview problems, in time polynomial in $N$
and exponential in $k$.
\par
First note that, with access to a \safeview oracle, the
standalone \SAsecureview\ problem can be easily solved in $O(2^k)$ time, by
calling the oracle for all $2^{k}$ possible subsets and outputting
the safe subset with minimum cost.
\par
Without access to a \safeview oracle, we first ``read'' relation $R$
using $N$ data supplier calls. Once $R$ is available, the simple
algorithm sketched below implements the \safeview oracle (i.e. tests
if a set $V$ of attributes is safe) and works in time $O(2^kN^2)$:
For a visible subset $V$, we look at all possible assignments to the
attributes in $I \setminus V$. For each input value we then check if
it leads to at least $\frac{\Gamma}{\prod_{a \in O \setminus V}
|\Delta_a|}$ different values of the visible output attributes in $O
\cap V$ ($\Delta_a$ is the domain of attribute $a$). This is a
necessary and sufficient condition for guaranteeing $\Gamma$
privacy, since by all possible $\prod_{a \in O \setminus V}
|\Delta_a|$ extensions of the output attributes, for each input,
there will be $\Gamma$ different possible output values (details
in Appendix~\ref{sec:ub_standalone_app}).

We mention here also that essentially the same algorithms (with same
upper bounds) can be used to output \emph{all} safe attribute sets
of a standalone module, rather than just one with minimum cost. Such
exhaustive enumeration will be useful in the following sections.



\vspace{2mm}

 \noindent \textbf{Remarks.} These results indicate
that, in the worse case, finding a minimal-cost safe attribute set
for a module may take time that is exponential in the number of
attributes. Note, however, that the number of attributes of a single
module is typically not large (often less than 10, see
\cite{myexpt}), so the computation is still feasible.
Expert knowledge of module designers, about the module's behavior and
safe attribute sets may also be exploited to speed up the
computation.
Furthermore,  a given module is often used in many
workflows. For example, sequence comparison modules, like BLAST or
FASTA, are used in many different biological workflows. We will see
that safe subsets for individual modules can be used as building
blocks for attaining privacy for the full workflow. The effort
invested in deriving safe subsets for a module is thus amortized
over all uses.

\section{All-Private Workflows}\label{sec:private-module}
We are now ready to consider workflows that consist of several
modules. We first consider in this section workflows where all
modules are {\em private} (called \emph{all-private workflows}).
Workflows with a mixture of private and public modules are then
considered in Section \ref{sec:public-module}.

As in Section~\ref{sec:standalone}, we want to find a safe
visible subset $V$ with minimum cost s.t. all the modules in the
workflow are $\Gamma$-workflow-private w.r.t. $V$ (see
Definition~\ref{def:workflow-privacy}). One option is to devise
algorithms similar to those described for standalone modules in
the previous section. However, the time complexity of those
algorithms is now exponential in the {\em total number of attributes
of all modules in the workflow} which can be as large as
$\Omega(nk)$, $n$ being the number of modules in the workflow and
$k$  the maximum number of attributes of a single module. To avoid
the exponential dependency on $n$, the number of modules in the
workflow,which may be large \cite{myexpt}, and to exploit the safe
attribute subsets for standalone modules, which may have been
already computed, we attempt in this section to assemble workflow
privacy guarantees out of standalone module guarantees. We first
prove, in Section~\ref{sec:privacy-private}, that this is indeed
possible. Then, in the rest of this section, we study the
optimization problem of obtaining a safe view with minimum cost.

Let $W$ be a workflow consisting of modules $m_1,\cdots, m_n$, where
$I_i, O_i$ denote the input and output attributes of $m_i$, $i \in
[1, n]$, respectively. We use below $R_i$ to denote the relation for
the \emph{standalone} module $m_i$. The relations $R = R_1 \Join R_2
\Join \cdots \Join R_n$, with attributes $A=\bigcup_{1=1}^n(I_i \cup
O_i)$, then describes the possible executions of $W$. Note that if
one of the modules in $W$ is not a one-to-one function then the
projection $\proj{I_i \cup O_i}{R}$ of the relation $R$ on $I_i \cup
O_i$ may be a subset of the (standalone) module relation $R_i$.

\par
In this section (and throughout the rest of the paper), for a set of
visible attributes $V \subseteq A$, $\overline{V} = A \setminus V$
will denote the hidden attributes. Further, $V_i = (I_i \cup O_i)
\cap V$ will denote the visible attributes for module $m_i$, whereas
$\overline{V_i} = (I_i \cup O_i) \setminus V_i$ will denote the
hidden attributes for $m_i$, for $i \in [1, n]$.

 \subsection{Standalone-Privacy vs. Workflow-Privacy}
\label{sec:privacy-private}

We show that if a set of visible attributes guarantees
$\Gamma$-standalone-privacy for a module, then if the module is
placed in a workflow where only a subset of those attributes is made
visible, then $\Gamma$-workflow-privacy is guaranteed for the module
in this workflow. In other words, in an all-private workflow, hiding
the union of the corresponding hidden attributes of the individual
modules guarantees $\Gamma$-workflow-privacy for all of
them\footnote{By Proposition~\ref{prop:superset}, this
also means that hiding any superset of this union would also be safe
for the same privacy guarantee.}. We formalize this next.

\begin{theorem}\label{thm:privacy-private}
Let $W$ be an all-private workflow with modules $m_1, \cdots, m_n$.
Given a parameter $\Gamma \geq 1$, let $V_i \subseteq (I_i \cup
O_i)$ be a set of visible attributes w.r.t which $m_i$, $i \in [1,
n]$, is $\Gamma$-standalone-private.
Then the workflow $W$ is $\Gamma$-private w.r.t the set of visible
attributes $V$ s.t.  $\overline{V} = \bigcup_{i = 1}^n
\widebar{V_i}$.
\end{theorem}

Before we prove the theorem, recall that $\Gamma$-standalone-privacy
of a module $m_i$ requires that for every input $\tup{x}$ to the
module, there are at least $\Gamma$ \emph{potential outputs} of
$\tup{x}$ in the possible worlds $\Worlds(R_i, {V_i})$ of the
standalone module relation $R_i$ w.r.t. $V_i$; similarly,
$\Gamma$-workflow-privacy of $m_i$ requires at least $\Gamma$
potential outputs of $\tup{x}$ in the possible worlds $\Worlds(R,
V)$ of the workflow relation $R$ w.r.t. $V$. Since $R = R_1 \Join
\cdots \Join R_n$, a possible approach to prove
Theorem~\ref{thm:privacy-private} may be to show that, whenever the
hidden attributes for $m_i$ are also hidden in the workflow $W$, any
relation $R'_i \in \Worlds(R_i, {V_i})$ has a corresponding relation
$R' \in \Worlds(R, V)$ s.t. $R'_i = \proj{I_i \cup O_i}{R'}$. If
this would hold, then for $\overline{V} = \bigcup_{i = 1}^n
\widebar{V_i}$, the set of possible outputs, for any input tuple
$\tup{x}$ to a module $m_i$, will remain unchanged.

Unfortunately, Proposition~\ref{prop:reduced_worlds} below shows
that the above approach fails. Indeed, $|\Worlds(R, V)|$ can be
significantly smaller than $|\Worlds(R_i, {V_i})|$
even for very simple workflows.
\begin{proposition}\label{prop:reduced_worlds}
There exist a workflow $W$ with relation $R$, a module $m_1$ in $W$
with (standalone) relation $R_1$, and a set of visible attributes
$V_1$ that guarantees both $\Gamma$-standalone-privacy and
$\Gamma$-workflow-privacy of $m_1$, such that the ratio of
$|\Worlds(R_1, V_1)|$ and $|\Worlds(R, V_1)|$ is {\em doubly
exponential} in the number of attributes 
of $W$.
\end{proposition}

\begin{proof}[Proof sketch]
To prove the proposition, we construct a simple workflow with two
modules $m_1, m_2$ connected as a chain. Both $m_1, m_2$ are one-one
functions with $k$ boolean inputs and $k$ boolean outputs (for
example, assume that $m_1$ is an identity function, whereas $m_2$
reverses the values of its $k$ inputs). The module $m_1$ gets
initial input attribute set $I_1$, produces $O_1 = I_2$ which is fed
to the module $m_2$ as input, and $m_2$ produces final attribute set
$O_2$. Let $V_1$ be an arbitrary subset of $O_1$ such that
$|\widebar{V_1}| = \log \Gamma$ 
(we assume that
$\Gamma$ is a power of 2). It can be verified that, $m_1$ as
a standalone module is $\Gamma$-standalone-private w.r.t. visible
attributes $V_1$ and both $m_1, m_2$, being one-one modules, are $\Gamma$-workflow-private
w.r.t. $V_1$.
\par
We show that the one-one nature of $m_1$ and $m_2$
restricts the size of $\Worlds(R, V_1)$ compared to that of
$\Worlds(R_1, V_1)$. Since both $m_1$ and $m_2$ are one-one
functions, the workflow $W$ also computes a one-one function. Hence
any relation $S$ in $\Worlds(R, V_1)$ has to compute a one-one
function as well. But when $m_1$ was standalone, any function
consistent with $V_1$ could be a member of $\Worlds(R_1, V_1)$. By a
careful computation, the ratio can be shown to be doubly exponential
in $k$ (details in 
Appendix~\ref{sec:app_proofs_private-module}).
\end{proof}

\par

Nevertheless, we  show below that 
for every input $\tup{x}$ of the module, the set of
its possible outputs, in these worlds, is exactly the same as that
in the original (much larger number of) module worlds. Hence privacy
is indeed preserved.

In proving
Theorem~\ref{thm:privacy-private}, our main technical tool is
Lemma~\ref{lem:main-private}, which states that given a set of
visible attributes $V_i$ of a standalone module $m_i$, the set of
possible outputs for \emph{every} input $\tup{x}$ to $m_i$ remains
unchanged when $m_i$ is placed in an all-private workflow, provided
the corresponding hidden attributes $\widebar{V_i}$ remains hidden
in the workflow.

Recall that $\Out_{x, m_i}$ and $\Out_{x, W}$ denote the possible
output for an input $\tup{x}$ to module $m_i$ w.r.t. a set of
visible attributes when $m_i$ is standalone and in a workflow $W$
respectively (see Definition~\ref{def:standalone-privacy} and
Definition~\ref{def:workflow-privacy}).

\begin{lemma}\label{lem:main-private}
Consider any module $m_i$ and any input $\tup{x} \in \proj{I_i}{R}$. If $\tup{y} \in \Out_{x, m_i}$
w.r.t. a set of visible attributes $V_i \subseteq (I_i \cup O_i)$,
then $\tup{y} \in \Out_{x, W}$ w.r.t.  $V_i \cup (A \setminus (I_i \cup O_i))$. 
\end{lemma}

The above lemma directly implies
Theorem~\ref{thm:privacy-private}:

\begin{proof}[of Theorem~\ref{thm:privacy-private}]
We are given that each module $m_i$ is $\Gamma$-standalone-private
w.r.t. $V_i$, i.e., $|\Out_{x, m_i}| \geq \Gamma$ for all input
$\tup{x}$ to $m_i$, for all modules $m_i$, $i \in [1, n]$ (see
Definition~\ref{def:standalone-privacy}). From
Lemma~\ref{lem:main-private}, this implies that for all input
$\tup{x}$ to all modules $m_i$, $|\Out_{x, W}| \geq \Gamma$ w.r.t
$V' = V_i \cup (A  \setminus (I_i \cup O_i))$. For this choice of
$V'$, $\widebar{V'} = A \setminus V' = (I_i \cup O_i) \setminus V_i$
$= \widebar{V_i}$ (because, $V_i \subseteq I_i \cup O_i \subseteq
A$). Now, using Proposition~\ref{prop:superset}, when the visible
attributes set $V$ is such that $\widebar{V} = \bigcup_{i = 1}^n
\widebar{V_i}$ $\supseteq \widebar{V_i} = \widebar{V'}$,
every module $m_i$ is $\Gamma$-workflow-private.
\end{proof}

To conclude the proof of Theorem~\ref{thm:privacy-private} we thus
only need to prove Lemma~\ref{lem:main-private}. For this, we use
the following auxiliary lemma.

\begin{lemma}\label{lem:out-x}
Let $m_i$ be a standalone module with relation $R_i$, let $\tup{x}$
be an input to $m_i$, and let $V_i \subseteq (I_i \cup O_i)$ be a
subset of visible attributes. If $\tup{y} \in \Out_{x, m_i}$ then
there exists an input
$\tup{x'} \in \proj{I_i}{R_i}$ to $m_i$
with output
$\tup{y'} = m_i(\tup{x'})$ such that $\proj{V_i \cap I_i}{\tup{x}} =
\proj{V_i \cap I_i}{\tup{x'}}$ and $\proj{V_i \cap O_i}{\tup{y}} =
\proj{V_i \cap O_i}{\tup{y'}}$.
\end{lemma}

The statement of the lemma can be illustrated with the module $m_1$
whose relation $R_1$ appears Figure~\ref{fig:m1}. Its visible
portion (for visible attributes $a_1,a_3,a_5$) is given in
Figure~\ref{fig:m1'}. Consider the input $\tup{x} = (0, 0)$ to $m_1$
and the output $\tup{y} = (1, 0, 0)$. For $V=\{a_1,a_3,a_5\}$,
$\tup{y} \in \Out_{x, m_1}$ (see Figure~\ref{tab:w3}).  This is
because there exists $\tup{x'} = (0, 1)$, s.t. $\tup{y'} = m_1(\tup{x'}) =
(1, 1, 0)$, and, $\tup{x, x'}$ and $\tup{y, y'}$ have the same
values of the visible attributes ($a_1$ and $\{a_3, a_5\}$
respectively). Note that $\tup{y}$ does not need to be the
actual output $m_1(\tup{x})$ on $\tup{x}$
 or even share the same values
of the visible attributes (indeed, $m_1(\tup{x}) = (0, 1, 1)$).
We defer the proof of Lemma~\ref{lem:out-x} to
Appendix~\ref{sec:app_proofs_private-module} and instead briefly
explain how it is used to prove Lemma~\ref{lem:main-private}.

\begin{proof}[Proof sketch of Lemma~\ref{lem:main-private}]
Let us fix a module $m_i$, an input $\tup{x}$ to $m_i$ and a
candidate output $\tup{y} \in \Out_{x, m_i}$ for $\tup{x}$ w.r.t.
visible attributes $V_i$. We already argued that, for $V = V_i \cup
(A \setminus (I_i \cup O_i))$, $\widebar{V} = A \setminus V = (I_i
\cup O_i) \setminus V_i = \widebar{V_i}$. We will show that $\tup{y}
\in \Out_{x, W}$ w.r.t. visible attributes $V$ by showing the
existence of a possible world $R' \in \Worlds(R, V)$, for
$\widebar{V} = \widebar{V_i}$,
s.t. 
$\proj{I_i}{\tup{t}} = \tup{x}$
and $\proj{O_i}{\tup{t}} = \tup{y}$
for some $\tup{t} \in R'$.
\par
 We start by replacing
module $m_i$ by a new module $g_i$ such that $g_i(\tup{x}) =
\tup{y}$ as required. But due to data sharing, other modules in the
workflow can have input and output attributes from $I_i$ and $O_i$.
Hence if we leave the modules $m_j, j \neq i$, unchanged, there may
be inconsistency in the values of the visible attributes, and the
relation produced by the join of the standalone relations of the
module sequence $\angb{m_1, \cdots, m_{i-1}, g_i, m_{i+1}, \cdots,
m_n}$ may not be a member of $\Worlds(R, V)$. To resolve this, we
consider the modules in a topological order, and change the
definition of all modules $m_1, \cdots, m_n$ to $g_1, \cdots, g_n$
(some modules may remain unchanged). In proving the above, the main
idea is to use \emph{tuple and function flipping} (formal definition
in the appendix). If a module $m_j$ shares attributes from $I_i$ or
$O_i$, the new definition $g_j$ of $m_j$ involves flipping the input
to $m_j$, apply $m_j$ on this flipped input, and then flipping the
output again to the output value. The proof shows that by
consistently flipping all modules, the visible attribute values
remain consistent with the original workflow relation and we get a
member of the possible worlds. (Details appear in Appendix
\ref{sec:app_proofs_private-module}).
\end{proof}

It is important to note that the assumption of all-private workflow is
crucial in proving Lemma~\ref{lem:main-private} -- if some of the
modules $m_j$ are public, we can not redefine them to $g_j$ (the
projection to the public modules should be unchanged - see
Definition \ref{def:workflow-privacy}) and we may not get a member
of $\Worlds(R, V)$. We will return to this point in Section~\ref{sec:public-module}
 when we consider workflows with a mixture of private and
public modules.

\subsection{The \secureview Problem} 
\label{sec:secureview-private}

We have seen above that one can assemble workflow privacy guarantees
out of the standalone module guarantees. Recall however that each
individual module may have several possible safe attributes sets
(see, e.g., Example~\ref{eg:sa-privacy}).  Assembling different sets
naturally lead to solutions with different cost.
\eat{
One may be tempted to believe that assembling optimal (cheapest)
safe attributes for the individual modules would lead to an optimal
safe attribute set for of the full workflow. However, the following
example should that this is not the case. } The following example
shows that assembling optimal (cheapest) safe attributes of the
individual modules may not lead to an optimal safe attributes set
for the full workflow. The key observation is that, due to data
sharing, it may be more cost effective to hide expensive shared
attributes rather than cheap non-shared ones
(though later we show that the problem remains NP-hard
even without data sharing).

\eat{ \scream{TOVA: I
CHANGED A BIT THE EXAMPLE SO THAT IT WILL CLEARER THAT OPTIMAL
SOLUTION CANNON BE SIMPLY OBTAINED BY ELIMINATING REDUNDANCY OF
HIDDEN ATTRIBUTES} }

\begin{example}\label{eg:1}
 Consider a workflow with $n+2$ modules, $m, m_1, \cdots, m_n, m'$.
 The module $m$ gets an input data item $a_1$, with cost $1$, and sends
 as output the same data item, $a_2$, with cost $1+\epsilon$, $\epsilon > 0$,
to all the $m_i$-s. Each $m_i$ then sends a data item $b_i$ to $m'$
with cost $1$. Assume that standalone privacy is preserved for
module $m$ if either its incoming or outgoing data is hidden and for
$m'$ if any of its incoming data is hidden. Also assume that
standalone privacy is preserved for each $m_i$ module if either its
incoming or its outgoing data is hidden. As standalone modules, $m$
will choose to hide $a_1$, each $m_i$ will choose to hide the outgoing
data item $b_i$, and $m'$ will choose to hide any of the $b_i$-s. 
The union of the optimal solutions for the standalone modules has
cost $n+1$. However, a lowest cost solution for
preserving workflow privacy is to hide $a_2$ and any one of
the $b_i$-s. This assembly of (non optimal) solutions for the individual
modules has cost $2+\epsilon$. In this case, the ratio of the costs of the union of 
standalone optimal solutions and the workflow optimal
solution is $\Omega(n)$.
\end{example}

This motivates us to define the combinatorial optimization problem
\secureview (for workflow secure view), which generalizes the
 \SAsecureview~ problem studied in Section~\ref{sec:standalone}.
 The goal of the \secureview problem is to choose, for each module, a safe set of attributes
 (among its possible sets of safe attributes) s.t. together the selected sets yield a
  minimal cost safe solution for the workflow.
We define this formally below. In particular, we consider the
following two variants of the problem, trading-off expressibility
and succinctness.

\medskip
\noindent \textbf{Set constraints.~~} The possible safe solutions
for a given module  can be given in the form of a list of 
hidden attribute sets. Specifically, we assume that we
are given, for each module $m_i$, $i \in [1, n]$, a list of pairs
$L_i= \langle (\widebar{I}^1_i,\widebar{O}^1_i),$
$(\widebar{I}^2_i,\widebar{O}^2_i)\ldots$
$(\widebar{I}^{l_i}_i,\widebar{O}^{l_i}_i) \rangle$. Each pair
$(\widebar{I}^j_i,\widebar{O}^j_i)$ in the list describes one
possible safe (hidden) solution for $m_i$:
\mbox{$\widebar{I}^j_i\subseteq I_i$} (resp. $\widebar{O}^j_i
\subseteq O_i$) is the set of input (output) attributes of $m_i$ to
be hidden in this solution. $l_i$ (the length of the list) is the
number of solutions for $m_i$ that are given in the list, and we use
below $\ell_{\max}$ to denote the length of the longest list, i.e.
$\ell_{\max} = \max_{i = 1}^n \ell_i$.

When the input to the \secureview problem is given in the above form
(with the candidate attribute sets listed explicitly) we call it
\emph{the \secureview problem \mbox{with set constraints}.}

\vspace{2mm}

 \noindent \textbf{Cardinality constraints.~~} Some
modules may have many possible candidate safe attribute sets.
Indeed, their number may be exponential in the number of attributes
of the module. This is illustrate by the following two simple
examples. \eat{ \scream{TOVA: I SLIGHTLY CHANGED THE FIRST EXAMPLE
PLEASE CHECK CORRECTNESS} }

\begin{example}
 First observe that in any one-one function with
$k$ boolean inputs and $k$ boolean outputs, hiding any $k$ incoming
or any $k$ outgoing attributes guarantees $2^{k}$-privacy. Thus
listing all such subsets requires a list of length $\Omega({2k
\choose k}) = \Omega(2^{k})$. Another example is majority function
which takes $2k$ boolean inputs and produces 1 if and only if the
number of one-s in the input tuple is $\geq k$. Hiding either $k+1$
input bits or the unique output bit guarantee $2$-privacy for
majority function, but explicitly listing all possible subsets again
leads to exponential length lists.
\end{example}

Note that, in both examples, the actual identity of the hidden input
(resp. output) attributes is not important, as long as sufficiently
many are hidden. Thus rather than explicitly listing all possible
safe sets we could simply say what combinations of numbers of hidden
input and output attributes are safe. This motivates the following
variant of the \secureview problem, called \emph{the \secureview
problem with cardinality constraints}: Here for every module $m_i$
we are given a list of pairs of numbers $L_i= \langle
(\alpha^1_i,\beta^1_i) \ldots (\alpha^{l_i}_i,\beta^{l_i}_i)
\rangle$, s.t. for each pair $(\alpha^j_i,\beta^j_i)$ in the list,
$\alpha^j_i \leq |I_i|$ and $\beta^j_i \leq |O_i|$. The
interpretation is that hiding any attribute set of $m_i$ that
consists of at least $\alpha^j_i$ input attributes and at least
$\beta^j_i$ output attributes, for some $j \in [1, \ell_i]$, makes
$m_i$ safe w.r.t the remaining visible attributes.

 To continue with
the above example, the list for the first module may consists of
$(k,0)$ and $(0,k)$, whereas the list for the second module consists
of $(k+1,0)$ and $(0,1)$.


It is easy to see that, for cardinality constraints, the lists are
of size at most quadratic in the number of attributes of the given
module (unlike the case of set constraints where the lists could be
of exponential length)\footnote{In fact, if one assumes that there
is no redundancy in the list, the lists become of at most of linear
size.}. In turn, cardinality constraints are less expressive than
set constraints that can specify arbitrary attribute sets. This will
affect  the complexity of the corresponding \secureview problems.\\

\noindent
\textbf{Problem Statement.~~}
Given an input in one of the two forms, a {\em feasible} safe subset
$V$ for the workflow, for the version with set constraints (resp.
cardinality constraints), is such that for each module $m_i$ $i \in
[1, n]$, $\widebar{V} \supseteq (\widebar{I}^j_i \cup
\widebar{O}^j_i)$ (resp. $|\widebar{V} \cap I_i| \geq \alpha^j_i$
and $|\widebar{V} \cap O_i| \geq \beta^j_i$) for some $j \in [1,
\ell_i]$. The goal of the \secureview problem is to find  a
safe set $V$ where $\cost(\widebar{V})$ is minimized.

\eat{
\vspace{2mm}

\scream{TOVA: I SLIGHTLY REPHRASED, PLEASE CHECK}
}
\subsection{Complexity results}\label{sec:complexity-private}
We present below theorems which give approximation algorithms and
matching hardness of approximation results of different versions of
the \secureview problem. The hardness results show that the problem
of testing whether the \secureview problem (in both variants) has a
solution with cost smaller than a given bound is  NP-hard even in
the most restricted case. But we show that certain approximations of
the optimal solution are possible.
Theorem~\ref{thm:cardinality-private} and \ref{thm:set-private}
summarize the results for the cardinality and set constraints
versions, respectively. For space constraints we only sketch the
proofs (full details appear in Appendix~\ref{sec:app_proofs_private-module}).

\begin{theorem}\label{thm:cardinality-private}
\textbf{(Cardinality Constraints)~~} There is an $O(\log
n)$-approximation of the \secureview problem with cardinality
constraints. Further,  this problem is $\Omega(\log n)$-hard
to approximate unless ${\rm NP}$ $\subseteq$ ${\rm
DTIME}$$(n^{O(\log \log n)})$, even if the maximum list size $\ell_{\max} = 1$, each
data has unit cost, and the values of $\alpha^j_{i}, \beta^j_{i}$-s
are 0 or 1.
\end{theorem}

\eat{
\scream{CHANGE NOTATION THROUGHOUT THE PROOFS}
}
\begin{proof}[Proof sketch]
The proof of the hardness result in the above theorem is by a
reduction from the set cover problem. The approximation is obtained
by randomized rounding a carefully written linear program (LP)
relaxation of this problem. A sketch is given below.

Our algorithm is based on rounding the fractional relaxation (called
the LP relaxation) of the integer linear program (IP) for this
problem presented in Figure~\ref{fig:ip}.

\begin{figure}[ht]
\centering \small Minimize $\sum_{b \in A} c_b x_b\quad$ subject to
\begin{eqnarray}
\sum_{j = 1}^{\ell_i} r_{ij} & \geq & 1\quad\forall i \in [1, n]\label{equn:IP1-1}\\
\sum_{b \in I_i} y_{bij} & \geq & r_{ij} \alpha^j_{i}\quad\forall i \in [1, n],~\forall j \in [1, \ell_i] \label{equn:IP1-6}\\%
\sum_{b \in O_i}z_{bij} & \geq & r_{ij} \beta^j_{i}\quad\forall i \in
[1, n],~\forall j \in [1, \ell_i]\label{equn:IP1-7}\\
\sum_{j = 1}^{\ell_i} y_{bij} & \leq & x_b,\quad\forall i \in [1, n], \forall b \in I_i\label{equn:IP1-2}\\
\sum_{j = 1}^{\ell_i} z_{bij} & \leq & x_b,\quad\forall i \in [1, n], \forall b \in O_i\label{equn:IP1-3}\\
y_{bij} & \leq & r_{ij},\quad\forall i \in [1, n],~\forall j \in [1, \ell_i],~\forall b \in I_i\nonumber\\
&&\label{equn:IP1-4}\\
z_{bij} & \leq & r_{ij},\quad\forall i \in [1, n],~\forall j \in [1, \ell_i],~\forall b \in O_i\nonumber\\
&&\label{equn:IP1-5}\\
x_b, r_{ij}, y_{bij}, z_{bij} & \in & \{0, 1\} \label{equn:IP1-8}
\end{eqnarray}
\vspace{-6mm} \caption{\mbox{IP for \secureview with cardinality
constraints} \label{fig:ip}} \vspace{-2mm}
\end{figure}

Recall that each module $m_i$ has a list 
$L_i = \{(\alpha^j_{i}, \beta^j_{i}): j\in [1, \ell_i]\}$, 
a feasible solution must ensure that for each $i\in [1, n]$, there
exists a $j \in [1, \ell_i]$ such that at least $\alpha^j_{i}$ input
data and $\beta^j_{i}$ output data of $m_i$ are hidden.

In this IP, $x_b = 1$ if data $b$ is hidden, and $r_{ij}=1$ if at
least $\alpha^j_{i}$ input data and $\beta^j_{i}$ output data of
module $m_i$ are hidden. Then, $y_{bij} = 1$ (resp., $z_{bij} = 1$)
if both $r_{ij} = 1$ and $x_b = 1$, i.e. if data $b$ contributes to
satisfying the input requirement $\alpha^j_{i}$ (resp., output
requirement $\beta^j_{i}$) of module $m_i$. Let us first verify that
the IP indeed solves the \secureview problem with cardinality
constraints. For each module $m_i$, constraint~(\ref{equn:IP1-1})
ensures that for some $j\in [1, \ell_i]$, $r_{ij} = 1$. In
conjunction with constraints~(\ref{equn:IP1-6}) and
(\ref{equn:IP1-7}), this ensures that for some $j\in [1, \ell_i]$,
(i) at least $\alpha^j_{i}$ input data of $m_i$ have $y_{bij} = 1$
and (ii) at least $\beta^j_{i}$ output data of $m_i$ have $z_{bij} =
1$. But, constraint~(\ref{equn:IP1-2}) (resp.,
constraint~(\ref{equn:IP1-3})) requires that whenever $y_{bij} = 1$
(resp., $z_{bij} = 1$), data $b$ be hidden, i.e. $x_b = 1$, and a
cost of $c_b$ be added to the objective. Thus the set of hidden data
satisfy the privacy requirement of each module $m_i$ and the value
of the objective is the cost of the hidden data. Note that
constraints~(\ref{equn:IP1-4}) and (\ref{equn:IP1-5}) are also
satisfied since $y_{bij}$ and $z_{bij}$ are 0 whenever $r_{ij} = 0$.
Thus, the IP represents the \secureview problem with cardinality
constraints. In Appendix~\ref{sec:cardinality-private} we show
that simpler LP relaxations of this problem without some of the
above constraints lead to unbounded and $\Omega(n)$ integrality gaps
showing that an $O(\log n)$-approximation cannot be obtained from
those simpler LP relaxations.

We round the fractional solution to the LP relaxation using
Algorithm~\ref{algo:lp1_round}. For each $j \in [1, \ell_i]$, let
$I^{min}_{ij}$ and $O^{min}_{ij}$ be the $\alpha^j_{i}$ input and
$\beta^j_{i}$ output data of $m_i$ with minimum cost. Then,
$B^{min}_i$ represents $I^{min}_{ij}\cup O^{min}_{ij}$ of minimum
cost.

\begin{algorithm}[h!t]
\caption{Rounding algorithm of LP relaxation of the IP given in
Figure~\ref{fig:ip},\\
\textbf{Input}: An optimal fractional
solution $\{x_b | b \in A\}$,\\
\textbf{Output}: A safe
subset $V$ for $\Gamma$-privacy of $W$.}
\begin{algorithmic}[1] \label{algo:lp1_round}
\STATE{Initialize $B = \phi$.} \STATE{For each attribute $b \in A$
($A$ is the set of all attributes in $W$), include $b$ in $B$ with
probability $\min\{1, 16x_b\log n\}$.}\label{step:rounding}
\STATE{For each module $m_i$ whose privacy requirement is not
satisfied by $B$, add $B^{min}_i$ to $B$.}\label{step:derandomize}
\STATE{Return $V = A \setminus B$ as the safe visible attribute.}
\end{algorithmic}
\end{algorithm}

The following lemma shows that step~\ref{step:rounding} satisfies
the privacy requirement of each module with high probability:

\begin{lemma}\label{lem:lp1_round}
Let $m_i$  be any module in workflow $W$. Then with probability at
least $1 - 2/n^2$, there exists a $j \in [1, \ell_i]$ such that
$|I^h_i| \geq \alpha^j_i$ and $|O^h_i| \geq \beta^j_i$.
\end{lemma}

\begin{proof}[Proof sketch]
The LP solution returns a probability distribution on $r_{ij}$, and
therefore on the pairs in list $L_i$. Let $p$ be the index of the
median of this distribution when list $L_i$ is ordered by both
$\alpha^j_{i}$ and $\beta^j_{i}$ values, as described above. Our proof
consists of showing that with probability $\geq 1-2/n^2$, $|I^h_i|
\geq \alpha_{ip}$ and $|O^h_i| \geq \beta_{ip}$.

Note that since $p$ is the median, the sum of $y_{bij}$ over all
incoming data of module $v_i$ in the LP solution must be at least
$\alpha_{ip}/2$ (from constraint~(\ref{equn:IP1-6})). Further,
constraint~(\ref{equn:IP1-4}) ensures that this sum is contributed
to by at least $\alpha_{ip}/2$ different input data, and
constraint~(\ref{equn:IP1-2}) ensures that $x_b$ for any input data
$b$ must be at least its contribution to this sum, i.e.
$\sum_{j}y_{bij}$. Thus, at least $\alpha_{ip}/2$ different input
data have a large enough value of $x_b$, and randomized rounding
produces a good solution. An identical argument works for the output
data of $m_i$ (details in the appendix).
\end{proof}
Since the above lemma holds for every module, by standard arguments,
the $O(\log n)$-approximation follows.
\end{proof}

We next show that the richer expressiveness of set constraints
increases the complexity of the problem.

\begin{theorem}\label{thm:set-private}
\textbf{(Set Constraints)~~}  The \secureview problem with set
constraints cannot be approximated to within a factor of
$\ell_{\max}^{\epsilon}$ for some constant $\epsilon > 0$ (also within a
factor of $\Omega(2^{\log^{1-\gamma}n})$ for all constant $\gamma >
0$) unless ${\rm NP}$ $\subseteq$ ${\rm DTIME}(n^{{\rm polylog~} n})$.
The hardness result holds even when the maximum list size
$\ell_{\max}$ is a (sufficiently large) constant, each data has unit
cost, and the subsets $\widebar{I}^j_i,\widebar{O}^j_i$-s have
cardinality at most 2. Finally, it is possible to get a factor
$\ell_{\max}$-approximation in polynomial time.
\end{theorem}

\begin{proof}[Proof sketch] When we are allowed to specify arbitrary subsets
for individual modules, we can encode a hard problem like
\emph{label-cover} which is known to have no poly-logarithmic
approximation given standard complexity assumptions. The
corresponding approximation is obtained by an LP rounding algorithm
which shows that a good approximation is still possible when the
number of specified subsets for individual modules is not too large.
Details can be found by Appendix~\ref{sec:set-private}.
\end{proof}

\eat{
Every data in a workflow is uniquely produced by a module, but can be fed as input to multiple modules:
this phenomenon was called \emph{data sharing}.
We say that a workflow has $\gamma$-bounded data sharing if the number of such modules
which takes a particular data item as input is bounded by $\gamma$.
Equivalently, if $I_i \rightarrow O_i$, $i = 1$ to $n$ are the functional dependencies
in relation $R$ of the workflow, every attribute $a$ of $R$ can appear on the left hand side
of at most $\gamma$ such dependencies.
}

The hardness proofs in the above two theorems use extensively data
sharing, namely the fact that an output attribute of a given module
may be fed as input to several other modules. Recall that a workflow
is said to have $\gamma$-bounded data sharing if the maximum number
of modules which takes a particular data item as input is bounded by
$\gamma$. In real life workflows, the number of modules where a data
item is sent is not very large. The following theorem shows that a
better approximation is possible when this number is bounded.

\begin{theorem}\label{thm:bounded-private}
\textbf{(Bounded Data Sharing)~~} There is a $(\gamma+1)$-approximation
algorithm for the \secureview problem (with both cardinality and set
constraints) when the workflow has $\gamma$-bounded data sharing. On
the other hand, the cardinality constraint version (and consequently
also the set constraint version) of the problem remain APX-hard
even when there is \emph{no} data sharing (i.e. $\gamma = 1$), each
data has unit cost, the maximum list size $\ell_{\max}$ is 2, and
the values of $\alpha^j_i, \beta^j_i$-s are bounded by 3.
\end{theorem}

\begin{proof}[Proof sketch]
The APX-hardness in the above theorem is obtained by a reduction
from vertex-cover in cubic graphs. This reduction also shows that
the NP-completeness of this problem does not originate from
data-sharing, and the problem is unlikely to have an exact solution
even without any data sharing. The $\gamma+1$-approximation is
obtained by a greedy algorithm, which chooses the least cost
attribute subsets for individual modules, and outputs the union of
all of them. Since any attribute is produced by a unique module and
is fed to at most $\gamma$ modules, in any optimal solution, a
single attribute can be used to satisfy the requirement of at most
$\gamma+1$ modules. This gives a $\gamma+1$-approximation. Observe
that when data sharing is not bounded, $\gamma$ can be $\Omega(n)$
and this greedy algorithm will not give a good approximation to this
problem.
\end{proof}

\section{Public Modules}\label{sec:public-module}
In the previous section we restricted our attention to workflows
where all modules are private. In practice, typical workflows use
also public modules. Not surprisingly, this makes privacy harder to
accomplish. In particular, we will see below that it becomes harder
to assemble privacy guarantees for the full workflow out of those
that suffice for component modules. Nevertheless a refined variant
of Theorem \ref{thm:privacy-private} can still be employed.

\subsection{\mbox{Standalone vs. Workflow Privacy (Revisited)}}
We have shown in Section \ref{sec:privacy-private}
(Theorem~\ref{thm:privacy-private}) that when a set of hidden
attributes guarantees $\Gamma$-standalone-privacy for a private
module, then the same set of attributes can  be used to guarantee
$\Gamma$-workflow-privacy {\em in an all-private network}.
Interestingly, this is no longer the case for workflows with public
modules. To see why, consider the following example.

\begin{example}\label{eg:public-challenge}
Consider a private module $m$ implementing a one-one function with
$k$ boolean inputs and $k$ boolean outputs. Hiding any $\log \Gamma$
input attributes guarantees $\Gamma$-standalone-privacy for $m$ even
if all output attributes of $m$ are visible. However, if $m$ gets
all its inputs from a public module $m'$ that computes
some constant function (i.e. $\forall \tup{x}, m'(\tup{x}) =
\tup{a}$, for some constant $\tup{a}$), then hiding $\log \Gamma$
input attributes no longer guarantees $\Gamma$-workflow-privacy of
$m$ -- this is because it suffices to look at the (visible) output
attributes of $m$ to know the value $m(\tup{x})$ for $x=\tup{a}$.

In an analogous manner, hiding any $\log \Gamma$ output attributes
of $m$, leaving all its input attributes visible, also guarantees
$\Gamma$-standalone-privacy of $m$. But if $m$ sends all its outputs
to another public module $m''$ that implements a one-one invertible
function, and whose output attributes happen to be visible, then for
any input $\tup{x}$ to $m$, $m(\tup{x})$ can be immediately inferred
using the inverse function of $m''$.
\end{example}

Modules that compute a constant function (or even one-one
invertible function) may not be common in practice. However, this simple
example illustrates where, more generally, the proof of
Theorem~\ref{thm:privacy-private} (or Lemma~\ref{lem:main-private})
fails in the presence of public modules: when searching for a
possible world that is consistent with the visible attributes, one
needs to ensure that the functions defined by the public modules
remain unchanged. So we no longer have the freedom of freely
changing the values of the hidden input (resp. output) attributes,
if those are supplied by (to) a public module.

One way to overcome this problem is to ``privatize" such problematic
public modules, in the sense that the name of the public module is
not revealed to users (either in the workflow specification or in
its execution logs).
Here we assume that once we rename a module the user loses all knowledge about it
(we discuss other possible approaches in the conclusion).
We refer to the public modules whose identity
is hidden (resp. revealed) as {\em hidden} ({\em visible}) public
modules.
 %
 %
 Observe that
now, since the identity of the hidden modules is no longer known to
the adversary, condition (2) in
Definition~\ref{def:pos-worlds-workflow} no longer needs to be
enforced for them, and a larger set of possible words can be
considered.
Formally,

\begin{definition}
\label{def:pos-worlds-workflow-public}(Definition
\ref{def:pos-worlds-workflow} revisited)
 Let  $P$ be a subset of the public modules, and, as before, let $V$ be
a set of the visible attributes. Then, the set of {\em possible
worlds} for the relation $R$ w.r.t. $V$ and $P$, denoted $\Worlds(R,
V, P)$, consists of all relations $R'$ over the same attributes as
$R$ that satisfy the functional dependencies in $F$
and where (1) $\proj{V}{R'} = \proj{V}{R}$, 
and (2) for every public module $m_i \in P$ and every tuple
$\tup{t'} \in R'$, $\proj{O_i}{\tup{t'}} =
m_i(\proj{I_i}{\tup{t'}})$.
\end{definition}

 %

 The notion of $\Gamma$-privacy for a workflow $W$,
with both private and public modules (w.r.t a set $V$ of visible
attributes and a set $P$ of visible public modules) is now defined as
before (Definition \ref{def:workflow-privacy}), except that the set
of possible worlds that is considered is the refined one from
Definition \ref{def:pos-worlds-workflow-public} above. Similarly, if
$W$ is \emph{$\Gamma$-private} w.r.t. $V$ and $P$, then we will call
the pair $(V,P)$ a \emph{safe subset} for $\Gamma$-privacy of $W$.

We can now show that, by making visible only public modules whose
input and output attribute values need not be masked, one can obtain
a result analogous to Theorem~\ref{thm:privacy-private}. Namely,
assemble the privacy guarantees of the individual modules to form
privacy guarantees for the full workflow. Wlog., we will assume that $m_1, m_2,
\cdots, m_K$ are the private modules and $m_{K+1}, \cdots, m_n$ are
the public modules in $W$. 

\begin{theorem}\label{thm:privacy-public}
Given a parameter
$\Gamma \geq 1$, let $V_i \subseteq (I_i \cup O_i)$, $i \in [1, K]$,
be a set of visible attributes w.r.t which the private module $m_i$
is $\Gamma$-standalone-private. Then the workflow $W$ is
$\Gamma$-private w.r.t the set of visible attributes $V$ and any set
of visible public modules $P \subseteq \{m_{K+1}, \cdots, m_n\}$,
s.t. $\overline{V} = \bigcup_{i = 1}^K \widebar{V_i}$ and all the
input and output attributes of modules in $P$ are visible and
\mbox{belong to $V$.}
\end{theorem}

\begin{proof}[Proof sketch]
The proof is similar to that of Thm.~\ref{thm:privacy-private}. Here
we additionally show  in a lemma analogous to
Lemma~\ref{lem:main-private} (see Lemma~\ref{lem:main-public} in
Appendix~\ref{sec:main-public}) that, if a public module $m_j, j \in
[K+1, n]$ is redefined to $g_j$, then $m_j$ is hidden. In other
words, the visible public modules in $P$ are never redefined and
therefore condition (2) in
Definition~\ref{def:pos-worlds-workflow-public} holds.
\end{proof}

\begin{example}
Consider a chain workflow with
three modules $m' \rightarrow m \rightarrow m''$, where $m'$ is a
public module computing a constant function, $m$ is a private module
computing a one-one function and $m''$ is another public module
computing an invertible one-one function. If we hide only a subset
of the input attributes of $m$, $m'$ should be hidden, thus $P =
\{m''\}$. Similarly, if we hide only a subset of the output
attributes of $m$, $m''$ should be hidden. Finally, if we hide a
combination of input and output attributes, both $m', m''$ should be
hidden and in that case $P = \phi$.
\end{example}

\subsection{The \secureview Problem (Revisited)}\label{sec:secureview-public}
The \secureview optimization problem in general
workflows is similar to the case of all-private workflows,
with an additional cost due to hiding (privatization)
of public modules: when a public module $m_j$ is hidden, the
solution incurs a cost $\pvt(m_j)$. Following the notation of
visible and hidden attributes, $V$ and $\widebar{V}$, we will denote
the set of hidden public modules by $\widebar{P}$. The total cost
due to hidden public modules is $\pvt(\widebar{P}) = \sum_{m_j \in
\widebar{P}} \pvt(m_j)$, and the total cost of a safe solution $(V,
P)$ is $\cost(\widebar{V}) + \pvt(\widebar{P})$. The definition of
the \secureview problem, with cardinality and set constraints,
naturally extends to this refined cost function and the goal is to
find a safe solution with minimum cost. This generalizes the
\secureview problem for all-private workflows where $\widebar{P} =
\phi$ (and hence $\pvt(\widebar{P}) = 0$).

\eat{
 In the following theorems, the term {\em general workflows} denotes workflows
with both public and private modules.
}


\medskip
\noindent
\textbf{Complexity Results} (details in Appendix~\ref{sec:app_proofs_public-module}).~~
In Section~\ref{sec:complexity-private} we showed that the \secureview
problem has an $O(\log n)$-approximation in an all-private workflow
even when the lists specifying cardinality requirements are $\Omega(n)$-long
and when the workflow has arbitrary data sharing. But, 
we show (in Appendix~\ref{sec:set-public}) by a reduction from the label-cover problem
that the cardinality constraints version in general workflows is
$\Omega(2^{\log^{1-\gamma}n})$-hard to approximate (for all constant $\gamma > 0$), and thus
unlikely to have any polylogarithmic-approximation.
\eat{
\begin{theorem}\label{thm:cardinality-public}
\textbf{Cardinality Constraints  (general workflows):~~}
Unless
${\rm NP} \subseteq {\rm DTIME}(n^{{\rm polylog~} n})$,
the \secureview problem with cardinality constraints in general workflows is
$\Omega(2^{\log^{1-\gamma}n})$-hard to approximate
for all constant $\gamma > 0$ even if the maximum size of the requirement lists is 1 and the individual
requirements are bounded by 1.
\end{theorem}
Theorem~\ref{thm:cardinality-public} is proved by a reduction from
the label-cover problem (see Appendix~\ref{sec:cardinality-public}).
} In contrast, the approximation factor for the set constraints
version remains the same and Theorem~\ref{thm:set-private} still
holds for general workflows by a simple modification to the proof.
However, $\gamma$-bounded data sharing no longer give a constant
factor approximation any more for a constant value of $\gamma$. By a
reduction from the set-cover problem, we prove in
Appendix~\ref{sec:bounded-public} that the problem is $\Omega(\log
n)$-hard to approximate even when the workflow has no data sharing,
and when the maximum size of the requirement lists and the
individual cardinality requirements in them are \mbox{bounded by 1.}

\vspace{-1mm}
\section{Conclusions}\label{sec:conclusions}
\vspace{-0.5mm} This paper proposes the use of provenance views for
preserving the privacy of module functionality in a workflow. Our
model motivates a natural optimization problem, \secureview , which
seeks to identify the smallest amount of data that needs to be
hidden so that the functionality of every module is kept private. We
give algorithms and hardness results that characterize the
complexity of the problem.

In our analysis, we assume that users have two sources of knowledge
about module functionality: the module name (identity) and the
visible part of the workflow relation. Module names are informative
for public modules, but the information is lost once the module
name is hidden/renamed.  Names of private modules are
non-informative, and users know only what is given in the workflow view.
However, if users have some additional prior knowledge about the
behavior of a private module, we may hide their identity
by renaming them, and then run our algorithms.

Our work suggests several promising directions for future research.
First, a finer privacy analysis may be
possible if one knows what {\em kind} of prior knowledge the user has on a
private module, e.g. the distribution of output values for
a specific input value, or knowledge about the types
and  names of input/output attributes (certain integers may be
illegal social security numbers, certain character sequences are
more likely to represent gene sequences than others, etc).
Our definitions and algorithms currently assume that all data values
in an attribute domain are equally possible, so the effect
of knowledge of a possibly non-uniform prior distribution on
input/output values should be explored.
Second, some additional sources of user knowledge on functionality of public modules
(e.g. types of attributes and connection with other modules)
may prohibit hiding their functionality using privatization (renaming),
and 
we would like to explore alternatives to privatization
to handle public modules. 
\eat{
One natural approach is to
\emph{propagate hidden attributes through public modules}: if
some input (output) attributes of a public module are hidden, then we hide
some or all of its output (input) attributes to ensure
privacy. However our initial results show that such approach does
not work, even in very simple workflows with restricted topologies
and limited data sharing.
}
A third direction to explore is an alternative model of
privacy.  As previously mentioned, standard mechanisms to guarantee differential
privacy (e.g. adding random noise to data values) do not seem to work for ensuring module privacy w.r.t.
provenance queries, and new mechanisms suitable to our application
 have to be developed. 
Other natural directions for future research include
considering \emph{non-additive cost functions}, in which some
attribute subsets are more useful than others, efficiently handling
\emph{infinite or very large domains} of attributes, and exploring
alternate objective functions, such as \emph{maximizing utility} of
visible data instead of minimizing the cost of hidden data. 

\medskip

\noindent
\textbf{Acknowledgements.}~~ 
%
S. B. Davidson, S. Khanna and S. Roy were supported in part by NSF-IIS Award 0803524;
 T. Milo was supported in part by NSF-IIS Award 1039376, the Israel Science Foundation, the US-Israel Binational
  Science Foundation and the EU grant MANCOOSI; and
   D. Panigrahi was supported in part by NSF-STC Award 0939370.

{
\small

}
\newpage
\appendix
\section{Proofs from Section 3}
\label{sec:app_proofs_standalone}

%
\subsection{Proof of Theorem~\ref{thm:communication-no-oracle}}

\noindent
\begin{proof}
We prove the theorem by a communication complexity reduction
from the set disjointness problem:
Suppose Alice and Bob own two subsets $A$ and $B$ of a universe $U$, $|U| = N$. To decide whether
they have a common element (i.e. $A \cap B \neq \phi$) takes $\Omega(N)$ communications \cite{NisanKBook}.
\par
We construct the following relation $R$ with $N+1$ rows for the module $m$.
$m$ has three input attributes: $a, b, id$ and one output attribute $y$.
The attributes $a, b$ and $y$ are boolean, whereas $id$ is in the range $[1, N+1]$.
The input attribute $id$ denotes the identity of every row in $R$ and takes value $i \in [1, N+1]$ for the $i$-th
row.
The module $m$ computes the AND function of inputs $a$ and $b$, i.e., $y = a \wedge b$.
\par
Row $i$, $i \in [1, N]$, corresponds to element $i \in U$.
In row $i$, value of $a$ is 1 iff $i \in A$; similarly,  value of  $b$ is 1 iff $i \in B$.
The additional $N+1$-th row has $a_{N+1} = 1$ and $b_{N+1} = 0$.
The standalone privacy requirement $\Gamma = 2$ and the goal is to check if visible attributes $V = \{id, y\}$
(with hidden attributes $\widebar{V} = \{a, b\}$) is safe for this privacy requirement.

\par
 Note that if there is a common element $i \in A \cap B$,
then there are two $y$ values in the table:
in the $i$-th row, the value of  $y = a \wedge b$ will be 1, whereas, in the $N+1$-th row
it is 0. Hence hiding $\widebar{V} = \{a, b\}$ will ensure the privacy requirement of $\Gamma = 2$
(every input $\tup{x}$ to $m$ can be mapped either to 0 or 1).
If there is no such $i \in A \cap B$, the value of $y$ in all rows $i \in [1, N+1]$ will be
zero which does not meet the privacy requirement $\Gamma = 2$.
Hence we need to look at $\Omega(N)$ rows to decide
whether $V = \{id, y\}$ is  safe.
\end{proof}

%
%
%
%

\subsection{Proof of Theorem~\ref{thm:computation-no-oracle}}
\noindent

In our reduction, $N = 2^{k-1}$. Hence, alternatively,
if $N$ is the number of tuples in the relation, there does not exists a
$poly(\log N)$ algorithm, unless P = NP.

\begin{proof}
We prove the theorem by a reduction from UNSAT:
Suppose we are given a boolean CNF formula $g$ on $\ell$ boolean variables $x_1, \cdots, x_\ell$.
The goal is to decide if \emph{no} assignment of $x_1, \cdots, x_\ell$ can satisfy $g$.
Given such an UNSAT instance, we build a relation $R$ with input attributes
$x_1, \cdots, x_\ell, y$ and
output attribute $z$, all of boolean domain (hence $k = \ell + 2$).
\par
The function of $m$  has a succinct description as follows:
$m(x_1, \cdots, x_\ell, y) = \neg g(x_1, \cdots, x_\ell) \wedge \neg y$ (i.e., NOR of $g(x_1, \cdots, x_\ell)$ and $y$).
Hence in the table $R$, \emph{implicitly}, we have two tuples for each assignment to the variables, $x_1, \cdots, x_\ell$:
if the assignment for $x_1, \cdots, x_\ell$ satisfies the formula $g$
then, for both $y = 0$ and $y = 1$, we have  $z = 0$.
Otherwise if the assignment does not satisfy the formula, for $y = 0$, we have $z = 1$,
and for $y = 1$, $z = 0$.
The privacy requirement is $\Gamma = 2$ and the goal is to decide if
visible attributes $V = \{x_1, \cdots, x_{\ell}\} \cup \{z\}$ is a safe subset where    hidden attributes $\widebar{V} = \{ y\}$.
\par
Note that if the formula $g$ is \emph{not} satisfiable, then it suffices to hide $y$ to get 2-privacy,
i.e. $V$ is safe. This is
because for \emph{every} satisfying assignment, there are two ways to complete $y$ value (one that
is the correct one and one that is the opposite).
On the other hand, if the function $g$ has at least one satisfying assignment,
for that assignment in $R$,regardless of the value of the hidden attribute $y$,
the output $z$ has to always be 0.
In that case $V$ is not a safe subset.
\end{proof}

%
%
\subsection{Proof of Theorem~\ref{thm:communication-oracle}}

\begin{proof}
Assume, for the sake of contradiction,
that an algorithm \Algo\ exists which uses $2^{o(k)}$
oracle calls. We will build an adversary that controls the \safeview oracle
and outputs answers to the queries consistent with a fixed function $m_1$ and
a dynamically changing function $m_2$ that depends on the set of queries asked.
The minimum cost of a safe subset for $m_1$ will be $3/2$ times that for 
(all definitions of) $m_2$, thereby proving the theorem.
\par
Consider a function 
with $\ell$ boolean input attributes in $I$, 
and one output attribute in $O$ 
where $\ell$ is even (i.e. $k = \ell +1$). 
The costs of all attributes in $I$ is 1, the cost of attribute $y$ in $O$
is $\ell$. 
We want to decide whether there exists a safe visible subset $V$ 
such that the cost of the hidden subset $\widebar{V}$ is
at most $C = \frac{\ell}{2}$, or all hidden subsets have
cost at least $\frac{3C}{2} = \frac{3\ell}{4}$. 
Hence, any such set 
$\widebar{V}$ can never include the output attribute. 

\par

The oracle behaves as follows:
\begin{itemize}
\item[(P1)] The oracle answers \Yes\ for every set $V$ of input attributes s.t. 
$|V| < \frac{\ell}{4}$  (i.e. $|\widebar{V}| > \frac{3\ell}{4}$), and
\item[(P2)] The oracle answers \No\ for every subset $V$ of 
input attributes s.t. $|V| \geq \frac{\ell}{4}$ 
(i.e. $|\widebar{V}| \leq \frac{3\ell}{4}$).
\end{itemize}

The functions $m_1$ and $m_2$ are defined as follows:
\begin{itemize}
\item $m_1$ returns 1 iff the total number of input attributes whose value is 1 
is at least $\frac{\ell}{4}$ (and otherwise 0),
\item $m_2$ has a special set $A$ such that $|A| = \frac{\ell}{2}$. It returns 1 iff 
the total number of input attributes whose value is 1 is at least $\frac{\ell}{4}$
and there is at least one input attribute outside $A$ whose value is 1 (and otherwise 0).
\end{itemize}
Note that while the cheapest safe subset for $m_1$ has cost greater than 
$\frac{3\ell}{4}$, $m_2$ has a safe subset $A$ where the cost of $\widebar{A}$ is $\frac{\ell}{2}$.
\par

It remains to show that the behavior of the oracle (i.e. properties (P1)
and (P2)) remains consistent with the definitions of $m_1$ and $m_2$.
We consider $m_1$ first. 
\begin{itemize}
 \item(P1) holds for $m_1$: 
	An all-0 $\widebar{V}$ and an all-1 $\widebar{V}$ 
	respectively imply 
	an answer of 0 and 1 independent of the assignment of $V$.
	\item(P2) holds for $m_1$: An all-1 $V$ implies an answer of 1 independent
	of the assignment of $\widebar{V}$.
\end{itemize}

Now, we consider $m_2$. 
\begin{itemize}
	\item (P1) holds for $m_2$: An all-0 $\widebar{V}$ and an all-1 $\widebar{V}$ 
	respectively imply 
	an answer of 0 and 1 independent of the assignment of $V$ or the definition of $A$
	(in the first case, number of 1 is $< \frac{\ell}{4}$ and in the second case the number of
	1 is $\geq \frac{3\ell}{4} > \frac{\ell}{4}$ and there is one 1 outside $A$ since
	$\frac{3\ell}{4} > \frac{\ell}{2}$).
	\item (P2) holds for $m_2$ as long as $V$ is not a subset of $A$, since
	an all-1 $V$ will imply an answer of 1 independent of the assignment of
	$\widebar{V}$. Therefore, such a query restricts the possible candidates
	of $A$, and discards at most ${3\ell/4\choose \ell/4}$ candidates of $A$.
	\end{itemize}
	
	Since there are ${\ell \choose \ell/2}$ possible definitions of $A$ overall, the 
	number of queries required to certify the absence of $A$ (i.e. certify that the 
	function is indeed $m_1$ and not $m_2$ with some definition of $A$) is at least
	\begin{equation*}
	\frac{{\ell \choose \ell/2}}{{3\ell/4\choose \ell/4}} = \prod_{i=0}^{\ell/2-1} \frac{\ell - i}{3\ell/4 - i} \geq (4/3)^{\ell/2} = 2^{\Omega(k)}.
	\end{equation*}

	Therefore, for a $2^{o(k)}$-restricted algorithm \Algo\ , there always 
	remains at least one subset $A$ defining a function $m_2$ that is consistent
	with all previous answers to queries. Hence after $2^{o(k)}$
calls, if the algorithm decides that there is a safe subset with cost $< C$, we choose
$m$ to be $m_1$; on the other hand, if it says that there is no such subset,
we set $m = m_2$ (with the remaining consistent subset of size $\frac{\ell}{2}$ as its special subset $A$). 
In both the cases the answer of the algorithm is wrong which
shows that there cannot be such an algorithm distinguishing these two cases
with $2^{o(k)}$ calls.	
\end{proof}

\medskip
\noindent
\textbf{Remark.~~} The above construction also shows that 
given a cost limit $C$,
deciding whether there exists a safe subset $V$
with cost at most $C$ or all safe
subsets have cost at least $\frac{3C}{2}$
requires $2^{\Omega(k)}$ oracle calls.
By adjusting the parameters in this construction, 
the gap can be increased to a factor of $\Omega(k^{1/3})$ from a constant.
More specifically, 
it can be shown that 
deciding whether there exists a safe subset with cost at most $C$,
or whether for all safe subsets the cost is at least $\Omega(k^{1/3}C)$
requires $2^{\Omega(k^{1/6})}$ calls to the \safeview\ oracle.

\subsection{Upper Bounds for Standalone Privacy}\label{sec:ub_standalone_app}

The lower bounds in Section~\ref{sec:lb_standalone} show that it is unlikely
to have optimal algorithms for the standalone \secureview\ problem 
with time complexity sub-exponential in $k$ and sub-linear in $N$.
Here we show that the naive algorithms for solving the \SAsecureview\  problem
have time complexity 
polynomial in $N$ and exponential in $k$.
It is important to note that the following upper bounds also hold
if our goal is to output \emph{all} safe subsets instead of deciding
whether there exists one within a cost limit. This will be useful in Section~\ref{sec:private-module}
and Section~\ref{sec:public-module} to guarantee workflow-privacy of private modules
where we will try to minimize the total cost of hidden attributes in the workflow
efficiently combining different options of hidden attributes that guarantee standalone-privacy
for the individual private modules.
\par
\begin{algorithm*}[ht]
\caption{Algorithm to find the minimum cost safe subset for a standalone module}
\begin{algorithmic}[1] \label{alg:standalone}
\FOR{every subset $V$ of $I \cup O$ such that $\cost(\widebar{V}) \leq C$}
    \FOR{every assignment
of the visible input attributes in $I \cap V$}
        \STATE{check if assignments of the hidden input attributes in $I \setminus V$
lead to at least $\frac{\Gamma}{\prod_{a \in O \setminus V} |\Delta_a|}$
different values of the visible output attributes in $O \cap V$($\Delta_a$ is the domain of attribute $a$)}
    \ENDFOR
\ENDFOR
\STATE{Output the subset $V$ having minimum cost $\cost(\widebar{V})$ satisfying the above condition.}
\end{algorithmic}
\end{algorithm*}

Let us first consider the \SAsecureview\  problem with a \safeview oracle.
Since there are at most $2^k$ possible subsets of attributes $V \subseteq I \cup O$ ($k = |I|+|O|$),
clearly $2^k$ calls to the \safeview oracle suffice to find a safe subset $V$
with minimum cost of $\widebar{V}$. 
showed in 
Theorem~\ref{thm:communication-oracle}. Hence with a \safeview oracle, the \SAsecureview\  problem
has $O(2^k)$ both communication and computation complexity.
\par
Next consider the upper bounds without a \safeview oracle. Clearly $N$ calls to the data supplier
suffices to read the entire relation $R$.
So we discuss standard computation complexity of the problem
once the relation $R$ is available (either by calls from the data supplier
or from a succinct representation).
Algorithm~\ref{alg:standalone} gives a simple algorithm to obtain a safe
subset with minimum cost. The following lemma
proves the correctness and time complexity of Algorithm~\ref{alg:standalone}.

\eat{
}
\begin{lemma}
Algorithm~\ref{alg:standalone} finds the optimal solution to the \safeview problem
when the relation $R$ of the module is given and runs in time $O(2^kN^2)$
where $k$ and $N$ are the number of attributes and the number of tuples in $R$ respectively.
\end{lemma}
\begin{proof}
First we prove the correctness of Algorithm~\ref{alg:standalone}.
When a subset $V \subseteq I \cup O$ satisfies the condition, each of these outputs can be extended to $\prod_{a \in O \setminus V} |\Delta_a|$
different outputs by all possible assignments of the hidden output attributes in $\prod_{a \in O \setminus V}$.
Together, this leads to $\Gamma$ different outputs for every assignment of the input attributes in $I$.
The algorithm returns the minimum cost subset among all subsets which are safe, and therefore
returns the optimal solution.
\par
Next we prove the time complexity.
Recall that $\delta = \max_{a \in I \cup O} |\Delta_a|$ denotes the maximum domain size of any attribute.
The number of different visible subsets $V$ of $I \cup O$ is $2^{k}$.
There are at most $\delta^{|I \cap V|}$ different assignments of $I \cap V$, but the relation $R$
with $N$ tuples can have at most $N$ of them. For each such assignment, there are
at most $\delta^{|I \setminus V|}$ assignments of $I \setminus V$.
However, to collect these outputs, we may need to scan the entire table which may take
$O(N)$ time.
Computing and checking whether the visible
output attributes give $\frac{\Gamma}{\prod_{a \in O \setminus V} |\Delta_a|}$ different values take time polynomial in $k$.
Hence total time complexity is $O(2^k N^2)$.
\par
The above algorithm can be implemented by standard SQL queries by going over all
possible subsets of attributes $V$, using a ``GROUP BY'' query on $I \cap V$,
and then checking if ``COUNT'' of every resulting tuples lead to the specified
number.
\end{proof}

This also shows that access to a \safeview\ 
oracle can improve the time complexity significantly ($2^k$ oracle calls suffice). 
The correctness of the algorithm
proved in the above theorem can be illustrated with Figure~\ref{fig:m1'}. Here $I = \{a_1, a_2\}$, $O = \{a_3, a_4, a_5\}$
and $V = \{a_1, a_3, a_5\}$. Consider assignment 0 to the visible input attribute $a_1$ in $I \cap V$.
The assignment of 0 to hidden input attribute $a_2 \in I \setminus V$ gives value $(0, 1)$ to visible output attributes
$a_3, a_5$ in $O \cap V$, whereas the assignment of 1 to $a_2$ gives $(1, 0)$.
For $\Gamma = 4$, we have $\frac{\Gamma}{|\Delta_{a_2}|} = \frac{4}{2} = 2$ different values of the visible output
attributes $a_3, a_5$ (here $|\Delta_{a_2}| = |\{0, 1\}| = 2$).
Each of these outputs can be extended to two different values by 0 and 1 assignment of the hidden output attribute $a_4$ in $O \setminus V$.
Hence each input $(0, 0)$ and $(0, 1)$ can be mapped to one of the four values $(0, 1, 0), (0, 1, 1), (1, 1, 0)$ and $(1, 1, 1)$.
The same can be verified for an assignment of 1 to the visible input attribute $a_1$ as well.
\par
\eat{
The time complexity analysis of Algorithm~\ref{alg:standalone} can be improved
to $O(2^{|I|}N^2)$ from $O(2^kN^2)$ by the same analysis. 
Further, when the relation $R$ contains all executions of the module $m$, then
$N$ is exponential in $|I|$. Hence 
the upper bound of $O(N) + O(2^kN^2)$ to get the relation $R$ from the data supplier (if needed)
and then to solve the standalone \secureview\ problem by Algorithm~\ref{alg:standalone}
can be polynomial in the lower bounds shown in 
Theorem~\ref{thm:communication-no-oracle} and \ref{thm:computation-no-oracle}.
}
\par
\eat{
//////////// SUDEEPA - UNNECESSARY, EXPLAINED IN THE MAIN BODY.
However, the number of input and output attributes $k = |I| + |O|$
of individual modules are expected to be small for all practical purposes.
In the following sections we will characterize workflow privacy of modules
in terms of their standalone privacy.
Even if we have access to a \safeview oracle for a workflow, total number of attributes
can be $\Omega(nk)$ leading to $O(2^{nk})$ oracle calls.
But if we have access to \safeview oracle for individual modules, using
our characterization the time complexity to find a feasible solution is
$O(n2^k)$, which is exponentially better in $n$.

\scream{This should be explained in a better way somewhere-- we should say that we can define
all the models defined here for the workflow, but the number of attributes in the workflow
can be $nk$, leading to $2^{nk}$ oracle calls, whereas here we need only $n2^k$.}
}

\section{Proofs from Section~4} 
\label{sec:app_proofs_private-module}

\subsection{Proof of Proposition~\ref{prop:reduced_worlds}}
We construct a simple workflow with two modules $m_1, m_2$ joined
back to back as a chain. Both $m_1, m_2$ are one-one functions with
$k$ boolean inputs and $k$ boolean outputs (for example, assume that
$m_1$ is an identity function, whereas $m_2$ reverses the values of
its $k$ inputs). The module $m_1$ gets initial input attribute set
$I_1$, produces $O_1 = I_2$ which is fed to the module $m_2$ as
input, and $m_2$ produces final attribute set $O_2$. Let $V_1$ be an
arbitrary subset of $O_1$ such that $|\widebar{V_1}| = \log \Gamma$
(for simplicity, we assume that $\Gamma$ is a power of 2). It can be
easily verified that, $m_1$ as a standalone module is
$\Gamma$-standalone-private w.r.t. visible attributes $V_1$ and both
$m_1, m_2$ are $\Gamma$-workflow-private w.r.t. visible attributes
$V_1$ (since $m_1, m_2$ are one-one modules).
\par
Next we discuss how the one-one nature of $m_1$ and $m_2$
restricts the size of $\Worlds(R, V_1)$ compared to that of
$\Worlds(R_1, V_1)$. Since both $m_1$ and $m_2$ are one-one
functions, the workflow $W$ also computes a one-one function. Hence
any relation $S$ in $\Worlds(R, V_1)$ has to compute a one-one
function as well. But when $m_1$ was standalone, any This in turn
implies that the projection $\proj{I_1 \cup O_1}{S}$ on  $I_1 \cup
O_1$ for any such relation $S$ has to be one-one as well, otherwise
$S$ cannot compute a one-one function ($S$ has to satisfy the
functional dependencies $I_1 \rightarrow O_1$ and $I_2 \rightarrow
O_2$). However, both $I_1$ and $O_2$ are visible and $S \in
\Worlds(R, V_1)$, i.e., $\proj{V_1}{S} = \proj{V_1}{R}$. Therefore
fixing the attribute values in $\proj{I_1 \cup O_1}{S}$ also fixes
the relation $S$. Hence the number of relations in $\Worlds(R, V_1)$
is exactly the same as the number of relations $S'$ over $I_1 \cup
O_1$ such that (1) $S'$ computes a one-one function from $I_1$ to
$O_1$, and, (2) $\proj{(I_1 \cup O_1) \cap V_1}{S'} = \proj{(I_1
\cup O_1) \cap V_1}{R_1}$. On the other hand, $\Worlds(R_1, V_1)$
will be all possible relations $S'$ on $I_i \cup O_i$ such that only
(2) holds.
\par
Let us first exactly compute $|\Worlds(R_1, V_1)|$.
Given an input to $m_1$, the visible output bits in $V_1$ are fixed;
however, the hidden output bits in $\widebar{V_1}$
can have arbitrary values. Since $|\widebar{V_1}| = \log \Gamma$,
any input to $m_1$ can be mapped to one of $\Gamma$ 
different outputs. There are $2^k$ different inputs to $m_1$, and any relation
$S' \in \Worlds(R_1, V_1)$
is an arbitrary combination of the mappings for individual inputs.
Hence $|\Worlds(R_1, V_1)| = \Gamma^{2^{k}}$.
\par
Next we compute  $|\Worlds(R, V_1)|$ which is the same as the number of
one-one mappings for the module $m_1$ with the same values of the visible bits.
Let us partition the set of $2^k$ different values of initial inputs to $m_1$
into $2^k/\Gamma$ groups, where all $\Gamma$ initial inputs in a
group produce the same values of visible intermediate attributes $V_1$.
Any relation $S \in \Worlds(R_1, V_1)$
has to map the input tuples in each such
group to $\Gamma$ \emph{distinct} intermediate tuples.
Hence $S$ must permute the $\Gamma$
intermediate tuples corresponding to a group of $\Gamma$
input tuples.
\par
Thus, the total number of relations in $\Worlds(R, V_1)$
is $(\Gamma!)^{2^k/\Gamma}\simeq ((2\pi \Gamma)^{1/2\Gamma}(\Gamma/e))^{2^k}$ by Stirling's approximation
(the input tuples in a group can map to one of $\Gamma!$ permutations, there are $2^k/\Gamma$ groups
which can map independently).
Hence
the ratio of $|\Worlds(R, V_1)|$ and $|\Worlds(R_1, V_1)|$
is
$\left(\frac{(2\pi \Gamma)^{1/2\Gamma}}{e}\right)^{2^k} < 1.4^{-2^k}$
for any $\Gamma \geq 2$
\footnote{For $x > 1$, $x^{1/x}$ is a decreasing function and
$\frac{e}{(\pi x)^{1/x}} \geq 1.4$}.

\subsection{Proof of Lemma~\ref{lem:out-x}}
\begin{proof}
If $\tup{y} \in \Out_{x, m_i}$ w.r.t. visible attributes $V_i$, then from Definition~\ref{def:standalone-privacy},
\begin{equation}
\exists R' \in \Worlds(R, {V_i}),~~\exists t' \in R'~~ s.t~~
  \tup{x} = \proj{I_i}{\tup{t'}} \wedge \tup{y}=\proj{O_i}{\tup{t'}}
  \end{equation}
  Further, from Definition~\ref{def:pos-worlds-standalone}, $R' \in \Worlds(R, {V_i})$
  only if $\proj{V_i}{R_i} = \proj{V_i}{R'}$. Hence there must exist a tuple $\tup{t} \in R_i$
  such that
  \begin{equation}
  \proj{V_i}{\tup{t}} = \proj{V_i}{\tup{t'}} \label{equn:out-x-1}
  \end{equation}
  Let $\tup{x'} = \proj{I_i}{\tup{t'}}$ and $\tup{y'} = \proj{O_i}{\tup{t'}}$, i.e. $\tup{y'} = m_i(\tup{x'})$.
  Then by definition of
  $\tup{x}, \tup{y}, \tup{x'}, \tup{y'}$ and from (\ref{equn:out-x-1}), $\proj{V_i \cap I_i}{\tup{x}} = \proj{V_i \cap I_i}{\tup{x'}}$
  and $\proj{V \cap O_i}{\tup{y}} = \proj{V \cap O_i}{\tup{y'}}$.
\end{proof}

\subsection{Proof of Lemma~\ref{lem:main-private}}

First we introduce some additional notations.
Recall that 
the relation $R$ for workflow $W$ is defined on the attribute set (or vector) $A$.
For a tuple $\tup{x}$, 
we will use $\tup{x}(Q)$ to denote that the tuple
$\tup{x}$ is defined on the attribute subset $Q \subseteq A$;
For an attribute $a \in Q$, $x[a]$ will denote the value of the attribute $a$ in $\tup{x}$,
i.e. $x[a] = \proj{a}{\tup{x}}$.
\par
Let $\tup{x}(P)$ be a tuple defined on an arbitrary attribute subset $P \subseteq A$
and $\tup{p}(Q), \tup{q}(Q)$ be two tuples defined on another arbitrary subset $Q \subseteq A$. Then,
the tuple $\tup{y} = \Flip_{\tup{p}, \tup{q}}(\tup{x})$ defined on attribute subset $P$ is defined as
\begin{displaymath}
y[a] = \left\{
    \begin{array}{ll}
        q[a] & \text{if } a \in Q \text{ and } x[a] = p[a]\\
        p[a] & \text{if } a \in Q \text{ and } x[a] = q[a]\\
        x[a] & \text{otherwise.}
    \end{array}
\right.
\end{displaymath}
Intuitively, if the input tuple $\tup{x}$ shares the same value of some common attribute $a \in P \cap Q$
with that of $\tup{p}$ (i.e. $x[a] = p[a]$),
the flip operation replaces the attribute value $x[a]$ by $q[a]$ in $\tup{x}$, whereas, if $x[a] = q[a]$,
it replaces the value $x[a]$ by $p[a]$. If $a \in P \setminus Q$, or, if  for some $a \in P \cap Q$, $x[a] \neq p[a]$
and $x[a] \neq q[a]$, $x[a]$ remains unchanged. If $x[a] = p[a] = q[a]$, then also the value of $x[a]$
remains unchanged.
\par
It is easy to see that $\Flip_{\tup{p, q}}(\Flip_{\tup{p, q}}(\tup{x})) = \tup{x}$. In the proof of the
lemma, we will also use the notion of \emph{function flipping}, which first flips the input on $\tup{p, q}$,
then applies the function on the flipped input, and then flips the output again on $\tup{p, q}$.
The formal definition is as follows.
\begin{definition}
Consider a module $m$ 
mapping attributes in $I$ to attributes in $O$.
Let $\tup{p}(P), \tup{q}(P)$ be two tuples defined on attribute subset $P \subseteq A$.
Then,
$\forall \tup{x}(X)$ defined on $X$, $\Flip_{m, \tup{p, q}}(\tup{x}) = \Flip_{\tup{p, q}}(m(\Flip_{\tup{p, q}}(\tup{x})))$.
\end{definition}
 Now we complete the proof
of Lemma~\ref{lem:main-private}.\\

\medskip
\noindent
\textbf{Proof of Lemma~\ref{lem:main-private}.}
\begin{proof}
If a module $m_i$ is $\Gamma$-workflow-private w.r.t. visible attributes $V_i$, then from
Proposition~\ref{prop:superset},
$m_i$ is also
$\Gamma$-workflow-private w.r.t. any $V_i' \subseteq V_i$ (or equivalently, $\widebar{V_i'} \supseteq \widebar{V_i}$).
Hence we will prove Lemma~\ref{lem:main-private} for $V = V_i$.
\par
Consider module $m_i$ with relation $R_i$, input tuple $\tup{x}$ and visible attribute subset $V_i$
 as stated in Lemma~\ref{lem:main-private}.
Let $\tup{y} \in \Out_{x, m_i}$.
We will prove $\tup{y} \in \Out_{x, W}$, by showing the existence of a relation $R' \in \Worlds(R, V)$ and a tuple $\tup{t'} \in R'$
such that $\tup{x} = \proj{I_i}{\tup{t'}} \wedge \tup{y}=\proj{O_i}{\tup{t'}}$ (ref. Definition~\ref{def:workflow-privacy}).
\par
Since $\tup{y} \in \Out_{x, m_i}$, by Lemma~\ref{lem:out-x},
there are $\tup{x'} \in \proj{V_i \cap I_i}{R_i}$, $\tup{y'} = m_i(\tup{x'})$ such that
\begin{equation}
\proj{V_i \cap I_i}{\tup{x}} = \proj{V_i \cap I_i}{\tup{x'}},
\proj{V_i \cap O_i}{\tup{y}} = \proj{V_i \cap O_i}{\tup{y'}} \label{equn:3}
\end{equation}
%
Let $\tup{p}(I_i \cup O_i), \tup{q}(I_i \cup O_i)$ 
be two tuples on $I_i \cup O_i$ where
\begin{displaymath}
p[\ell] = \left\{
    \begin{array}{ll}
        x[\ell] & \text{if } \ell \in I_i\\
        y[\ell] & \text{if } \ell \in O_i
    \end{array}
\right.\text{and }
q[\ell] = \left\{
    \begin{array}{ll}
        x'[\ell] & \text{if } \ell \in I_i\\
        y'[\ell] & \text{if } \ell \in O_i.
    \end{array}
\right.
\end{displaymath}
(Recall that $I_i \cap O_i = \phi$). Hence $\Flip_{\tup{p, q}}(\tup{x'}) = \tup{x}$
and $\Flip_{\tup{p, q}}(\tup{y'}) = \tup{y}$. It should be notes that $\tup{x}, \tup{x'}$
and $\tup{y}, \tup{y'}$ have the same values on visible attribute
subsets $I_i \cap V$ and $O_i \cap V$ respectively. So $\tup{p}$ and $\tup{q}$ only differ
on the hidden attributes. Therefore, for any two tuples $\tup{w}, \tup{z}$, if
$\Flip_{\tup{p, q}}(\tup{w}) = \tup{z}$, then $\tup{w}$ and $\tup{z}$ will also only differ
on the hidden attributes and their visible attribute values are the same.
\par
For each $j \in [1, n]$, we define
$g_j = \Flip_{m_j, \tup{p}, \tup{q}}$. 
Then the desired relation $R' \in \Worlds(R, {V})$ is obtained by collecting executions
of the workflow where every module $m_i$ is replaced by module $g_i$, $i \in [1, n]$.
So we need to show that (i) there is a tuple $\tup{t} \in S$, such that $\proj{I_i}{\tup{t}} = \tup{x}$
and $\proj{O_i}{\tup{t}} = \tup{y}$, and, (ii) $R' \in \Worlds(R, {V})$.

\par
\textbf{(i):~~}
To show the existence of such a tuple $\tup{t} \in R'$, it suffices to show that $g_i(\tup{x}) = \tup{y}$,
since then for any tuple $\tup{t} \in R'$, if $\proj{I_i}{\tup{t}} = \tup{x}$, then $\proj{O_i}{\tup{t}} = \tup{y}$.
We claim that $g_i$ maps $\tup{x}$ to $\tup{y}$ as desired. This holds since
$g_i(\tup{x}) =$ $ \Flip_{m_i, \tup{p, q}}(\tup{x})$ $= \Flip_{\tup{p, q}}(m_i(\Flip_{\tup{p, q}}(\tup{x})))$
$ =  \Flip_{\tup{p, q}}(m_i(\tup{x'}))$ $= \Flip_{\tup{p, q}}(\tup{y'}) = \tup{y}$.
\par
\textbf{(ii):~~}
Since every $g_j$ is a function, $R'$ satisfies all functional dependencies $I_i \rightarrow O_i$, $i \in [1, n]$.
Hence to prove $R' \in \Worlds(R, V)$,
it suffices to show that, for the same \emph{initial inputs} in $R$ and $R'$, the values of
all the visible attributes in $R$ and $R'$ are the same.
Let $I_0$ be the set of initial inputs to workflow $W$.
We need to show that for any two tuples $\tup{t} \in R$ and $\tup{t'} \in R'$ on attribute set $A$,
if $\proj{I_0}{\tup{t}} = \proj{I_0}{\tup{t'}}$, then $\tup{t}, \tup{t'}$ also have the same
values on the visible attributes $V$, i.e.,
$\proj{V}{\tup{t}} = \proj{V}{\tup{t'}}$.
\par
Let us fix any arbitrary tuple $\tup{p}$ on input attributes $I_0$.
Let us assume, wlog., that the modules $m_1, \cdots, m_n$ (and corresponding $g_1, \cdots, g_n$)
are ordered in a topological sorted order in the DAG $W$.
Since $I_0$ is essentially the key attributes of relation $R$ or $R'$,
there are two unique tuples $\tup{t} \in R$ and $\tup{t'} \in R'$
such that $\proj{I_0}{\tup{t}} = \proj{I_0}{\tup{t'}} = \tup{p}$.
Note that any intermediate or final attribute $a \in A \setminus I_0$ belongs to $O_j$, for a unique $j \in [1, n]$
(since for $j \neq \ell$, $O_j \cap O_{\ell} = \phi$).
We prove by induction on $j$ that the values of the visible attributes $O_j \cap V$
are the same for $\tup{t}$ and $\tup{t'}$ for every $j$. Together with the fact that the values of the attributes in $I_0$
are the same in $\tup{t}, \tup{t'}$, this shows that $\proj{V}{\tup{t}} = \proj{V}{\tup{t'}}$.
\par
Let $\tup{c_{j, f}}, \tup{c_{j, g}}$ be the values of input attributes $I_j$ and $\tup{d_{j, f}}, \tup{d_{j, g}}$
be the values of output attributes $O_j$ of module $m_j$ in $\tup{t} \in R$ and $\tup{t'} \in R'$
respectively on initial input attributes $\tup{p}$
(i.e. $\tup{c_{j, f}} = \proj{I_j}{\tup{t}}$, $\tup{c_{j, g}} = \proj{I_j}{\tup{t'}}$,
$\tup{d_{j, f}} = \proj{O_j}{\tup{t}}$ and $\tup{d_{j, g}} = \proj{O_j}{\tup{t'}}$).
Then, we prove that
$\tup{d_{j, g}} = \Flip_{\tup{p, q}}(\tup{d_{j, f}})$.

From (\ref{equn:3}), $\tup{x}, \tup{x'}$, 
and $\tup{y}, \tup{y'}$ 
have the same values of the visible attributes.
Therefore the tuples $\tup{p}$ and $\tup{q}$ only differ in hidden attributes.
Then if the above claim is true, for every $j$, $\tup{d_{j, g}}$ and $\tup{d_{j, f}}$
are the same on the visible attributes $O_i \cap V$.
Equivalently, $\tup{t}$ and $\tup{t'}$ have the same values of visible attributes $V$
as desired.
\par
Note
that if the inductive hypothesis holds for all $j' < j$,
then $\tup{c_{j, g}} = \Flip_{\tup{p, q}}(\tup{c_{j, f}})$, since
the modules are listed in a topological order. Thus,
\begin{eqnarray*}
\tup{d_{j, g}} & = & g_j(\tup{c_{j, g}}) = \Flip_{m_j, \tup{p, q}}(\Flip_{\tup{p, q}}(\tup{c_{j, f}}))\\
& = & \Flip_{\tup{p, q}}(m_j(\Flip_{\tup{p, q}}(\Flip_{\tup{p, q}}(\tup{c_{j, f}}))))\\
& = & \Flip_{\tup{p, q}}(m_j(\tup{c_{j, f}})) = \Flip_{\tup{p, q}}(\tup{d_{j, f}}).
\end{eqnarray*}

Hence the hypothesis $\tup{d_{j, g}} = \Flip_{\tup{p, q}}(\tup{d_{j, f}})$ also holds for module $m_j$.
This completes the proof of this lemma. 
\end{proof}
\subsection{Proof of Theorem~\ref{thm:cardinality-private}}\label{sec:cardinality-private}
In \secureview problem with cardinality constraint, as stated in Section~\ref{sec:secureview-private},
every module $m_i$, $i \in [1, n]$, has a requirement list of pair of numbers 
$L_i = \{(\alpha^j_{i}, \beta^j_{i}): \alpha^j_{i} \leq |I_i|, \beta^j_{i} \leq |O_i|, j \in [1, \ell_i]\}$.
The goal is to select a safe subset of attributes $V$ with minimum cost $\cost(\widebar{V})$,
such that for every $i\in [1, n]$,
at least $\alpha^j_{i}$ input attributes and $\beta^j_{i}$ output attributes
of module $m_i$ are hidden for some $j \in [1, \ell_i]$. 
\par
In this section, we prove Theorem~\ref{thm:cardinality-private}.
First we give an $O(\log n)$-approximation algorithm,
and then show that 
this problem is
$\Omega(\log n)$-hard under standard complexity-theoretic 
assumptions, even if the cost of hiding each data is identical,
and the requirement list of every module in the workflow 
contains exactly one pair of numbers with values 0 or 1.


\subsubsection{$O(\log n)$-Approximation Algorithm}

Our algorithm is based on rounding the fractional relaxation (called the LP relaxation) 
of the integer linear program (IP) for this problem
presented in Figure~\ref{fig:ip}.

One can write a simpler IP for this problem, 
where the summations are removed from constraints~(\ref{equn:IP1-2})
and (\ref{equn:IP1-3}), and constraints~(\ref{equn:IP1-4})
and (\ref{equn:IP1-5}) are removed altogether. To see the 
necessity of these constraints, consider the LP relaxation
of the IP, obtained by replacing  constraint~(\ref{equn:IP1-8}) with
$x_b, r_{ij},  y_{bij}, z_{bij} \in [0, 1]$. 

Suppose constraints~(\ref{equn:IP1-4}) and 
(\ref{equn:IP1-5}) were missing from the IP, and therefore
from the LP as well.
For a particular $i \in [1, n]$, it is possible that 
a fractional solution to the LP has $r_{ij} = 1/2$ for
two distinct values $j_1$ and $j_2$ of $j$, where
$\alpha_{ij_1} > \alpha_{ij_2}$ and 
$\beta_{ij_1} < \beta_{ij_2}$. 
But constraint~(\ref{equn:IP1-6})
(resp., constraint~(\ref{equn:IP1-7})) can now be satisfied 
by setting $y_{bij_1} = y_{bij_2} = 1$ (resp., 
$z_{bij_1} = z_{bij_2} = 1$) for $\alpha_{ij_1}/2$ 
input data (resp., $\beta_{ij_2}/2$ output data).
However, $(\alpha_{ij_1}/2, \beta_{ij_2}/2)$
might not satisfy the privacy requirement for $i$, 
forcing an integral solution to hide some data $b$ with 
$x_b = 0$. This will lead to an unbounded integrality 
gap.

Now, suppose constraints~(\ref{equn:IP1-6}) and 
(\ref{equn:IP1-7}) did not have the summation.
For a particular $i \in [1, n]$, it is possible that 
a fractional solution to the LP has $r_{ij} = 1/\ell_i$ for
all $j \in [1, \ell_i]$. Constraint~(\ref{equn:IP1-6})
(resp., constraint~(\ref{equn:IP1-7})) can then be satisfied 
by setting $y_{bij} = 1/\ell_i$ for all $1\leq \ell_i$, 
for $\max_j \{\alpha^j_{i}\}$ distinct
input data (resp., $\max_j \{\beta^j_{i}\}$ 
distinct output data).
Correspondingly, $x_b = 1/\ell_i$ for those data $b$.
If all the $\alpha^j_{i}$s and $\beta^j_{i}$s for 
different $j \in [1, \ell_i]$ have similar values, it
would mean that we are satisfying the privacy constraint
for $m_i$ paying an $\ell_i$ fraction of the cost of an
integral solution. This can be formalized to yield an
integrality gap of $\max_i\{\ell_i\}$, which could be
$n$. Introducing the summation in the LP precludes this
possibility.

\eat{
We round the fractional solution to the LP relaxation using 
Algorithm~\ref{algo:lp1_round}. For each $j \in [1, \ell_i]$,
let $I^{min}_{ij}$ and $O^{min}_{ij}$ be the
$\alpha^j_{i}$ input and $\beta^j_{i}$ output data of $m_i$
with minimum cost. Then, $B^{min}_i$ represents 
$I^{min}_{ij}\cup O^{min}_{ij}$ of minimum cost.

\begin{algorithm}[h!t]
\caption{Rounding algorithm of LP relaxation of the IP given in Figure~\ref{fig:ip},~~~~~~~~~
\textbf{Input}: An optimal fractional solution $\{x_b | b \in A\}$,~~~~~~~~~
\textbf{Output}: A safe subset $V$ for the workflow $W$ for privacy requirement $\Gamma$
}
\begin{algorithmic}[1] \label{algo:lp1_round}
\STATE{Initialize $B = \phi$.}
\STATE{For each attribute $b \in A$ ($A$ is the set of all attributes in $W$), 
include $b$ in $B$ with probability $\min\{1, 16x_b\log n\}$.}\label{step:rounding}
\STATE{For each module $m_i$ whose privacy requirement is not satisfied by $B$, 
add $B^{min}_i$ to $B$.}\label{step:derandomize}
\STATE{Return $V = A \setminus B$ as the safe visible attribute.}
\end{algorithmic}
\end{algorithm}

}

\noindent{\bf Analysis.} 
We can assume wlog that the requirement list $L_i$ for each module $m_i$ 
is non-redundant, i.e. for all $1\leq j_1\not= j_2 \leq \ell_i$, either
$\alpha_{ij_1} > \alpha_{ij_2}$ and $\beta_{ij_2} < \beta_{ij_1}$, or
$\alpha_{ij_1} < \alpha_{ij_2}$ and $\beta_{ij_2} > \beta_{ij_1}$.
We can thus assume that for each module $m_i$, 
the list $L_i$ is sorted in increasing order 
on the values of $\alpha^j_{i}$ and in decreasing order on the values
of $\beta^j_{i}$.
The following lemma shows that step~\ref{step:rounding} satisfies the 
privacy requirement of each module with high probability.\\ 

\medskip
\noindent

\textsc{Lemma~\ref{lem:lp1_round}}
\emph{
Let $m_i$  be any module in workflow $W$. Then with probability at least $1 - 2/n^2$,
there exists a $j \in [1, \ell_i]$ such that 
$|I^h_i| \geq \alpha^j_{i}$ and $|O^h_i| \geq \beta^j_{i}$.
}

\medskip
\noindent
\begin{proof}
Given a fractional solution to the LP relaxation, 
let $p \in [1, \ell_i]$ be the index corresponding 
to the \emph{median} $(\alpha^j_{i}, \beta^j_{i})$,
satisfying $\sum_{j = 1}^{p - 1} r_{ij} < 1/2$ and 
$\sum_{j = 1}^{p} r_{ij} \geq 1/2$. 
We will show that after step~\ref{step:rounding},
at least 
$\alpha_{ip}$ input data and $\beta_{ip}$ output data
is hidden with probability at least $1-2/n^2$ 
for module $m_i$.

We partition the set of data $A$ 
into two sets: the set of data deterministically 
included in $B$, $B^{det}= \{b: x_b \geq 1/16\log n\}$,
and the set of data probabilistically rounded,
$B^{prob} = A \setminus B^{det}$. Also, let 
$B^{round} = B \setminus B^{det}$ be the set of 
data that are actually hidden among $B^{prob}$. For each module $m_i$, 
let $I^{det}_i = B^{det} \cap I_i$ and
$O^{det}_i = B^{det} \cap O_i$ be the set of
hidden input and output data in $B^{det}$ respectively. 
Let the size of these sets be $\alpha^{det}_i = |I^{det}_i|$
and $\beta^{det}_i = |O^{det}_i|$. Also, let 
$I^{prob}_i = B^{prob}\cap I_i$ and
$O^{prob}_i = B^{prob}\cap O_i$. Finally, let 
$I^{round}_i = B^{round}\cap I_i$ and
$O^{round}_i = B^{round}\cap O_i$. 
We show that for any module $m_i$,
$|I^{round}_i|\geq \alpha_{ip} -\alpha^{det}_i$
and $|O^{round}_i|\geq \beta_{ip} -\beta^{det}_i$ 
with probability at least $1 - 1/n^2$.

First we show that %
$\sum_{b \in I^{prob}_i} x_b \geq (\alpha_{im} -\alpha^{det}_i)/2$. 
Constraint~(\ref{equn:IP1-6}) implies that $\sum_{b \in I_i} y_{bij} \geq r_{ij} \alpha^j_{i}$, while constraint~(\ref{equn:IP1-4}) ensures that $\sum_{b \in I^{det}_i} y_{bij} \leq r_{ij} \alpha^{det}_i$. Combining these, we have
\begin{equation}
\sum_{b\in I^{prob}_i} y_{bij} \geq r_{ij}(\alpha^j_{i}-\alpha^{det}_i).\label{equn:analysis}
\end{equation} 
From constraint~(\ref{equn:IP1-2}), we have
\begin{equation*}
\sum_{b \in I^{prob}_i} x_b 
\geq \sum_{b \in I^{prob}_i} \sum_{j = 1}^{\ell_i} y_{bij}
\geq \sum_{j = p}^{\ell_i} \sum_{b \in I^{prob}_i} y_{bij}.
\end{equation*}
Then, from Eqn.~(\ref{equn:analysis}),
\begin{equation*}
\sum_{b \in I^{prob}_i} x_b 
\geq \sum_{j = p}^{\ell_i} r_{ij} (\alpha^j_{i} - \alpha^{det}_i)
\geq (\alpha_{ip} - \alpha^{det}_i) \sum_{j = p}^{\ell_i} r_{ij}.
\end{equation*}
Finally, using constraint~(\ref{equn:IP1-1}), we conclude that
\begin{equation}
\sum_{b \in I^{prob}_i} x_b \geq \frac{\alpha_{ip} - \alpha^{det}_i}{2}.\label{equn:alpha}
\end{equation}
Similarly, since $\sum_{j = 1}^p r_{ij} \geq 1/2$
and the list $L_i$ is sorted in decreasing order of $\beta^j_{i}$, it follows that
\begin{equation}
\sum_{b \in O^{prob}_i} x_b \geq \frac{\beta_{ip} -\beta^{det}_i}{2}. \label{equn:beta}
\end{equation}
Next we show that
$|I^{round}_i| \geq \alpha_{ip} -\alpha^{det}_i$ with probability
$\geq 1 - 1/n^2$.
Each $b \in B^{prob}$ is independently included in $B^{round}$ 
with probability $16x_b \log n$. 
Hence, by Eqn.~(\ref{equn:alpha}), 
\begin{equation*}
E[|I^{round}_i|] = \sum_{b \in I^{prob}_i} 16x_b \log n
\geq 8(\alpha_{ip} - \alpha^{det}_i) \log n.
\end{equation*}  
Using Chernoff bound\footnote{If $X$ is sum of 
independent boolean random variables with $E[X] = \mu$, then 
$\Pr[X \leq \mu(1 - \epsilon)] \leq e^{\frac{-\mu \epsilon^2}{2}}$(see, for instance, \cite{MR}).},  
$|D_{round} \cap I_i| \leq \alpha_{ip} - \alpha^{det}_i$ 
with probability at most $1/n^2$. Similarly,
using Eqn.~(\ref{equn:beta}), $|O^{round}_i| \leq \beta_{ip} - \beta^{det}_i$ 
with probability at most $1/n^2$.
The lemma follows by using union bound over the failure probabilities.
%
\end{proof}
We get the following corollary from Lemma~\ref{lem:lp1_round}, which proves
the approximation result in Theorem~\ref{thm:cardinality-private}.
\begin{corollary}
Algorithm~\ref{algo:lp1_round} 
gives a feasible safe subset $V$ with expected cost $O(\log n)$ 
times the optimal. 
\end{corollary}
\begin{proof}
Using union bound over the set of $n$ modules in the above lemma, 
we conclude that with probability at least $1 - 2/n$,
the solution produced by the rounding algorithm after 
step~\ref{step:rounding} is feasible.
By linearity of expectation, the cost of the rounded solution
at this stage is
at most $16\log n$ times that of the LP solution,
and therefore $O(\log n)$ times that of the optimal cost.
If all modules are not satisfied after step~\ref{step:derandomize},
the cost of the data added to $B$ in step~\ref{step:derandomize} 
(by greedily picking the best option $B^{min}_i$ for individual module)
is at most $O(n)$ times the optimal. However, 
this happens with probability at most $2/n$; thus
the expected total cost of the final solution $V$ produced by this algorithm 
remains $O(\log n)$ times the optimal cost. Further, the solution $V$
returned by the algorithm is always a safe subset.
\end{proof}
%

\subsubsection{$\Omega(\log n)$-Hardness}

The following theorem shows that Algorithm~\ref{algo:lp1_round} 
produces an optimal answer upto a constant factor and proves
the hardness result in Theorem~\ref{thm:cardinality-private}.

We give a reduction from the minimum set cover problem 
to this version of the \secureview problem where $\ell_{\max} = 1$
and each data has unit cost. 
Since set cover 
is hard to approximate within a factor of $o(\log n)$ unless ${\rm NP} \subseteq {\rm DTIME}(n^{O(\log \log n)})$
\cite{Feige98, LY94},
the hardness result of the \secureview problem with cardinality constraints follows 
under the same assumption.

An instance of the set cover problem consists of
an input universe $U = \{u_1, u_2, \ldots, u_n\}$, 
and a set of its subsets
${\bf S} = \{S_1, S_2, \ldots, S_M\}$, i.e.
each $S_i\subseteq U$. The goal is to find a set of subsets
${\bf T}\subseteq {\bf S}$ of minimum size (i.e. $|{\bf T}|$
is minimized) subject to the constraint that 
$\cup_{S_i\in {\bf T}} S_i = U$.

We create an instance 
of the 
\secureview problem with workflow $W$, where $W$
has a module $m_i$ corresponding to each element
$u_i\in U$, and an extra module $z$ (in addition to the dummy source and sink modules $s$ and $t$)
We now express the connections between modules in the workflow $W$.
There is a
single incoming edge $e_z$ 
from source module $s$ to $z$ (for initial input data), a set of 
edges $\{e_{ij}: S_i\ni u_j\}$ from $z$ to each $f_j$ (for intermediate data),
and a single outgoing edge $e_j$ from each $f_j$ (for final output data).
 to the sink node $t$. 
 The edge $e_z$ 
 uniquely carries data
$b_s$, and each edge $e_j$ uniquely carries data $b_j$
for $j\in [1, n]$. 
All edges $\{e_{ij}: j\in S_i\}$ carry the same data
$a_i$, $i \in [1, M]$.
\par
The privacy requirement 
for $z$ is any single data $a_i$ carried by one of its
outgoing edges, while that for each $f_j$ is any 
single data $a_i$ carried by one of its incoming edges (i.e. $S_i \ni u_j$).
In other words, $L_z = \{(0, 1)\}$,
$L_{j} = \{(1, 0)\}$.
Hence only the intermediate data, $\{a_i: i \in [1, M]\}$,
can be hidden;
the cost of hiding each such data is 1.
Note that the maximum list size $\ell_{\max}$ is 1
and the individual cadinality requirements are bounded by 1.
\par
If the minimum set cover problem has a cover of size $k$, hiding the data
corresponding to the subsets selected in the cover produces a solution of
cost $k$ for this instance of the \secureview problem. Conversely, if a
solution to the \secureview problem hides a set of $k$ data in 
$\{a_i: i \in [1, M]\}$,
selecting the corresponding sets produces a cover of $k$ sets.
Hence the \secureview problem with cardinality constraint is $\Omega(\log n)$-hard
to approximate. 

\subsection{Proof of Theorem~\ref{thm:set-private}}\label{sec:set-private}
We now consider the \secureview problem with set constraints. 
Here the input requirement lists $L_i$-s are given as a list of pair of subsets of input and output attributes:
$L_i = \{(\widebar{I}^j_i, \widebar{O}^j_i): j \in [1, \ell_i], \widebar{I}^j_i \subseteq I_i, \widebar{O}^j_i \subseteq O_i\}$,
for every $i \in [1, n]$ (see Section~\ref{sec:secureview-private}). 
The goal is find a safe subset $V$ with minimum cost of hidden attributes $\cost(\widebar{V})$
such that for every $i \in [1, n]$,  
$\widebar{V} \supseteq (\widebar{I}^j_i\cup \widebar{O}^j_i)$ for some $j \in [1, \ell_i]$.
\par
%
Recall that $\ell_{\max}$ denotes the maximum size of the requirement list of a module. 
Now we prove Theorem~\ref{thm:set-private}. First we show that the \secureview problem with
set constraints is $\ell_{\max}^{\epsilon}$-hard to approximate, and then we give an $\ell_{\max}$-approximation
algorithm for this problem.

\subsubsection{$\ell_{\max}$-Approximation Algorithm}\label{sec:set-private-approx}
%

Here we give an $\ell_{\max}$-approximation algorithm for the \secureview
problem with set constraints as claimed in Theorem~\ref{thm:set-private}.
The algorithm rounds the solution given by LP relaxation of the following integer program:
\begin{center}
Minimize $\sum_{b \in A} c_b x_b$\quad subject to 
\begin{align}
& \sum_{j = 1}^{\ell_i} r_{ij} \geq 1 & \forall i \in [1, n]\label{equn:IP2-1}\\
& x_b \geq r_{ij} & \forall b \in \widebar{I}^j_i \cup \widebar{O}^j_i,~\forall i \in [1, n] \label{equn:IP2-2}\\%
& x_b, r_{i, j} \in \{0, 1\} &\label{equn:IP2-3}
\end{align}
\end{center}
The LP relaxation is obtained by changing Eqn.~(\ref{equn:IP2-3}) to
\begin{equation}
x_b, r_{ij} \in [0, 1].  \label{equn:LP2-3}
\end{equation}
The rounding algorithm includes all attributes $b \in A$
such that $x_b \geq 1/\ell_{\max}$ to the hidden attribute set $\widebar{V}$.
The corresponding visible attribute subset $V$ is output as a solution.
\par

Next we discuss the correctness and approximation ratio of the rounding algorithm. 
Since the maximum size of a requirement list is $\ell_{\max}$, 
for each $i$, there exists a $j = j(i)$ such that in the solution of the LP,
$r_{ij} \geq 1/\ell_i \geq 1/\ell_{\max}$. Hence there exists at least one
$j \in [1, \ell_i]$ such that $\widebar{I}^j_i \subseteq \widebar{V}, \widebar{O}^j_i \subseteq \widebar{V}$.
Since $\cost(\widebar{V})$ is most $\ell_{\max}$ times the cost of LP solution, 
this algorithm gives an $\ell_{\max}$-approximation.

\subsubsection{$\ell_{\max}^\epsilon$-Hardness}
The hardness result in 
Theorem~\ref{thm:set-private} is obtained by a reduction from the 
\emph{minimum label cover} problem~\cite{ABS+93}. An instance of
the minimum label cover problem consists of a bipartite graph 
$H = (U, U', E_H)$, a label set $L$,
and a non-empty relation $R_{uw} \subseteq L \times L$ for each edge
$(u, w) \in E_H$. A feasible solution is a label assignment to the 
vertices, $A : U \cup U'  \rightarrow 2^{L}$, 
such that for each
edge $(u, w)$, there exist $\ell_1 \in A(u), \ell_2 \in A(w)$ 
such that $(\ell_1, \ell_2) \in R_{uw}$.
The objective is to find a feasible solution that minimizes
$\sum_{u \in U \cup U'}|A(u)|$.

Unless ${\rm NP} \subseteq {\rm DTIME}(n^{{\rm polylog~} n})$,
the label cover problem is $|L|^\epsilon$-hard to approximate
for some constant $\epsilon > 0$~\cite{ABS+93, Raz98}.
The instance of the \secureview problem in the reduction will have
$\ell_{max} = |L|^2$. Theorem~\ref{thm:set-private} follows immediately.

Given an instance of the label cover problem as defined above, 
we create an instance 
of the \secureview problem 
by constructing a workflow $W$ (refer to Figure~\ref{fig:lc1}).
For each edge $(u, w) \in E_H$,
there is a module $x_{uw}$ in $W$. In addition, $W$ has another 
module $z$. 

\par

As shown in Figure~\ref{fig:lc1}, the input and output attributes
of the modules are as follows:
(i) $z$ has a single incoming edge with the initial input data item $b_z$,
(ii) $z$ has $(|U| + |U'|) \times L$ output attributes $b_{u, \ell}$, for every
$u \in U \cup U'$ and every $\ell \in L$. Every such attribute $b_{u, \ell}$
is sent to all $x_{uw}$ where $(u, w) \in E_H$ (see Figure~\ref{fig:lc1}).
Hence every $x_{uw}$ has $2L$ input attributes: $\{b_{u, \ell}~:~ \ell \in L \}$
$\bigcup \{b_{w, \ell'}~:~ \ell' \in L \}$.
(iii) there is a single outgoing edge from each $x_{uw}$ carrying 
data item $b_{uw}$ (final output data).
\eat{

As shown in Figure~\ref{fig:lc1}, there are three
kinds of edges in $V$: 
(i) a single incoming edge $e_z$ 
to $z$ corresponding to initial input,
(ii) a set of $|2L|$ parallel edges $\{e_{uw, u, \ell}~:~ \ell \in L\}$
$\cup \{e_{uw, w, \ell}~:~ \ell \in L\}$ 
from $z$ to each $x_{uw}$ indexed by one of $u$ or $w$ and a label 
$\ell\in L$, and,
(iii) a single outgoing edge $e_{uw}$ from each $x_{uw}$ 
for final outputs.
Edge $e_z$ and $e_{uw}$, for each $(u, w)\in E$, 
uniquely carry data $b_z$ and $b_{uw}$ respectively, whereas
data $b_{u, \ell}$ is carried by all edges $e_{uw, u, \ell}$. 
}
\begin{figure}[h]
  \begin{center}
    \includegraphics[scale=.35]{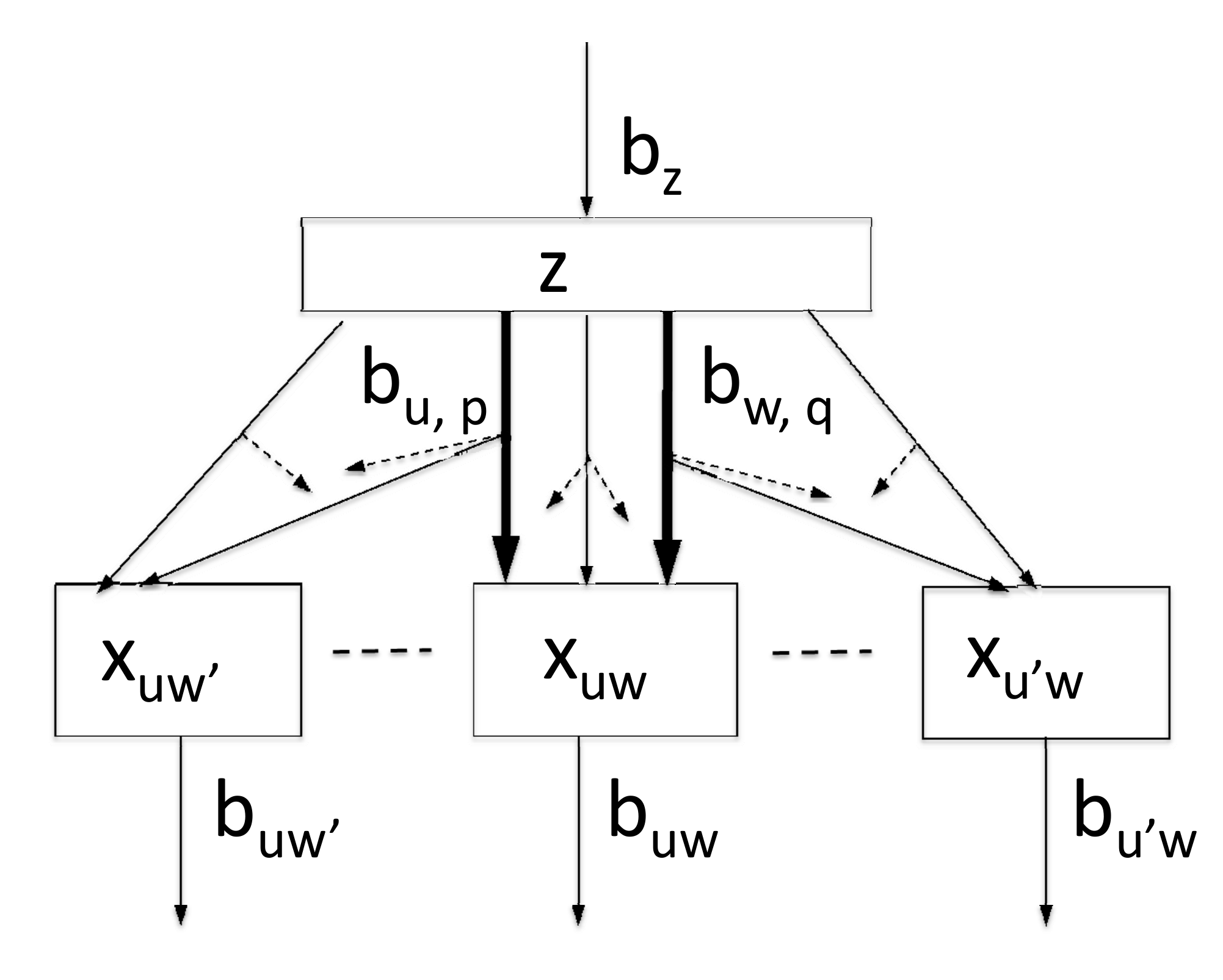}
  \end{center}
  \caption{Reduction from label cover: 
  bold edges correspond to $(p, q) \in R_{uw}$, dotted edges
  are for data sharing.}
  \label{fig:lc1}
\end{figure}
The cost of hiding any data is 1.
The requirement list of $z$ contains singleton subsets of each intermediate data $b_{u, \ell}$,
i.e., $L_z = \{(\phi, \{b_{u, \ell}\}): u \in U \cup U', \ell \in L\}$. 
The list of $x_{uw}$, for each $(u, w)\in E$, contains pairs of data 
corresponding to the members of the relation $R_{uw}$, i.e
$L_{uw} = \{(\phi, \{b_{u, \ell_1}, b_{w, \ell_2}\}): (\ell_1, \ell_2) \in R_{uw}\}.$

The following lemma 
proves the correctness of the reduction.
\begin{lemma}\label{lem:lc}
The label cover instance $H$ has a solution of cost $K$ iff 
the \secureview instance $W$ has a solution of cost $K$.
\end{lemma}
\begin{proof}
Let $A: U \cup U' \rightarrow 2^L$ be a solution of the label cover instance $H$
with total cost $K = \sum_{u \in U \cup U'} |A(u)|$.
We create a solution $\widebar{V}$ for the \secureview instance $W$ as follows:
for each $u \in U \cup U'$, and $\ell \in L$, add $b_{u, \ell}$ to hidden attributes $\widebar{V}$
iff $\ell \in A(u)$. We claim that $\widebar{V}$ is a feasible solution for $G$.
For each $u \in U \cup U'$, $A(u)$ is non-empty; hence, the requirement of $z$
is trivially satisfied.
Since $A$ is a valid label cover solution, for each $(u, w) \in E_H$, there exists
$\ell_1\in A(u)$ and $\ell_2 \in A(w)$ such that $(\ell_1, \ell_2) \in R_{uw}$.
Hence for the same pair $(\ell_1, \ell_2)$, $(b_{u, \ell_1}, b_{w, \ell_2}) \in L_{x_{uw}}$,
and both $b_{u, \ell_1}, b_{w, \ell_2} \in \widebar{V}$. 
This satisfies the requirement for all modules $x_{uw}$ in $W$. 

Conversely, let $\widebar{V}$ be a solution of the \secureview instance $W$, where $|\widebar{V}| = K$. 
Note that $\widebar{V}$ can include only the intermediate data.
For each $u \in U \cup U'$, we define $A(u) = \{\ell | b_{u, \ell} \in \widebar{V}\}$.
Clearly, $\sum_{u \in U \cup U'} |A(u)| = K$. 
For each $x_{uw} \in V$, the requirement of $x_{uw}$ is satisfied by $\widebar{V}$;
hence there exist $\ell_1, \ell_2 \in L$ such that $b_{u, \ell_1}, b_{w, \ell_2} \in \widebar{V}$.
This implies that for each edge $(u, w) \in E_H$, there exist $\ell_1\in A(u)$
and $\ell_2 \in A(w)$, where $(\ell_1, \ell_2) \in R_{uw}$, thereby proving
feasibility.
\end{proof} 
 
\noindent
\textbf{Remark.} 
If $N = |U|+|U'|$, the label cover problem is also known to be
$\Omega(2^{\log^{1-\gamma}N})$-hard to approximate
for all constant $\gamma > 0$, unless 
${\rm NP} \subseteq {\rm DTIME}(n^{{\rm polylog~} n})$~\cite{ABS+93, Raz98}. 
Thus, the \secureview problem with set constraints is
$\Omega(2^{\log^{1-\gamma}n})$-hard to approximate as well, 
for all constant $\gamma > 0$, under the same complexity
assumption.

%
%


\subsection{Proof of Theorem~\ref{thm:bounded-private}}\label{sec:bounded-private}
The \secureview problem 
becomes substantially easier to approximate if the workflow
has \emph{bounded data sharing}, i.e. when every data $d$ produced by some module
is either a final output data or is an input data to at most $\gamma$
other modules. 
Though the problem remains NP-hard
even with this restriction, Theorem~\ref{thm:bounded-private}
 shows that it is possible to approximate it within 
a constant factor 
when $\gamma$ is a constant. 


\eat{
\begin{theorem}\label{thm:bounded-private}
There is a $(\gamma+1)$-approximation algorithm for both versions of 
the \secureview problem when each data acts as input to at most $\gamma$ modules.
On the other hand, both versions of the problem remain APX-hard 
even when there is \emph{no}
data sharing (i.e. $\gamma = 1$), each data has unit cost, and 
$\ell_{\max}$ is 2.
\end{theorem}

}

\noindent
First, we give a $(\gamma+1)$-approximation algorithm for the \secureview problem with
set constraints, where each data is shared by at most $\gamma$ edges.
This also implies an identical approximation factor for
the cardinality version.
Then, we show that the cardinality version of the problem is APX-hard,
i.e. there exists a constant $c > 1$ such that it is NP-hard to obtain
a $c$-approximate solution to the problem. The set version is therefore
APX-hard as well.

\subsubsection{$(\gamma + 1)$-Approximation Algorithm}
Recall that the input to the problem includes
a requirement list 
$L_i = \{(\widebar{I}^j_i, \widebar{O}^j_i): j \in [1, \ell_i], \widebar{I}^j_i \subseteq I_i, \widebar{O}^j_i \subseteq O_i\}$
for each module $v_i$.
Let 
$(\widebar{I}^{j*}_{i}, \widebar{O}^{j*}_{i})$
be a minimum cost pair for module $v_i$,
i.e. 
$\cost(\widebar{I}^{j*}_{i} \cup \widebar{O}^{j*}_{i}) = \min_{j=1}^{\ell_i} \cost(\widebar{I}^j_i \cup \widebar{O}^j_i)$. 
The algorithm greedily 
chooses $\widebar{I}^{j*}_{i} \cup \widebar{O}^{j*}_{i}$
for each module $v_i$, i.e. the set of hidden data 
$\widebar{V} = \bigcup_{1\leq i \leq n}(\widebar{I}^{j*}_{i} \cup \widebar{O}^{j*}_{i})$.

Note that
each intermediate data is an input to at most $\gamma$ modules. In
any optimal solution, assume that each terminal module of
every hidden edge carrying this data pays its cost. 
Then,
the total cost paid by the modules is at most $\gamma+1$ times 
the cost of the optimal
solution. On the other hand, the total cost paid by any module is 
at least the cost of the edges incident on the module that are hidden
by the algorithm. Thus, the solution of the algorithm has cost
at most $\gamma+1$ times the optimal cost.

A very similar greedy algorithm with the same approximation factor
can be given to the \secureview problem with cardinality constraints.

\subsubsection{APX-Hardness}
We give a reduction from the \emph{minimum vertex cover} 
problem in cubic graphs
to the \secureview problem with cardinality constraints. 
An instance of the vertex cover problem consists of
an undirected graph $G'(V', E')$. The goal is to find 
a subset of vertices $S \subseteq V'$
of minimum size $|S|$ such that each edge $e \in E'$ 
has at least one endpoint in $S$.

Given an instance of the vertex cover problem, we 
create an instance of the \secureview problem 
$W$
(see Figure~\ref{fig:vc}).
For each edge $(u, v) \in E'$, there is a module $x_{uv}$ in $W$;
also there is a module $y_v$ for each vertex $v \in V'$. 
In addition to these, $W$ contains a single module $z$. 
\par
Next we define the edges in the workflow $W$; since there is no 
data sharing, each edge corresponds to a unique data item
and cost of hiding each edge is 1.
For each $x_{uv}$, 
there is a single incoming edge 
(carrying initial input data)
and two outgoing edges $(x_{uv}, y_u)$ and $(x_{uv}, y_v)$.
There is an outgoing edge $(y_v, z)$ from every $y_v$ (carrying final output data).
Finally, there is an outgoing edge from $z$ for final output data item. 
\par
\begin{figure}[ht]
  \begin{center}
    \includegraphics[scale=.35]{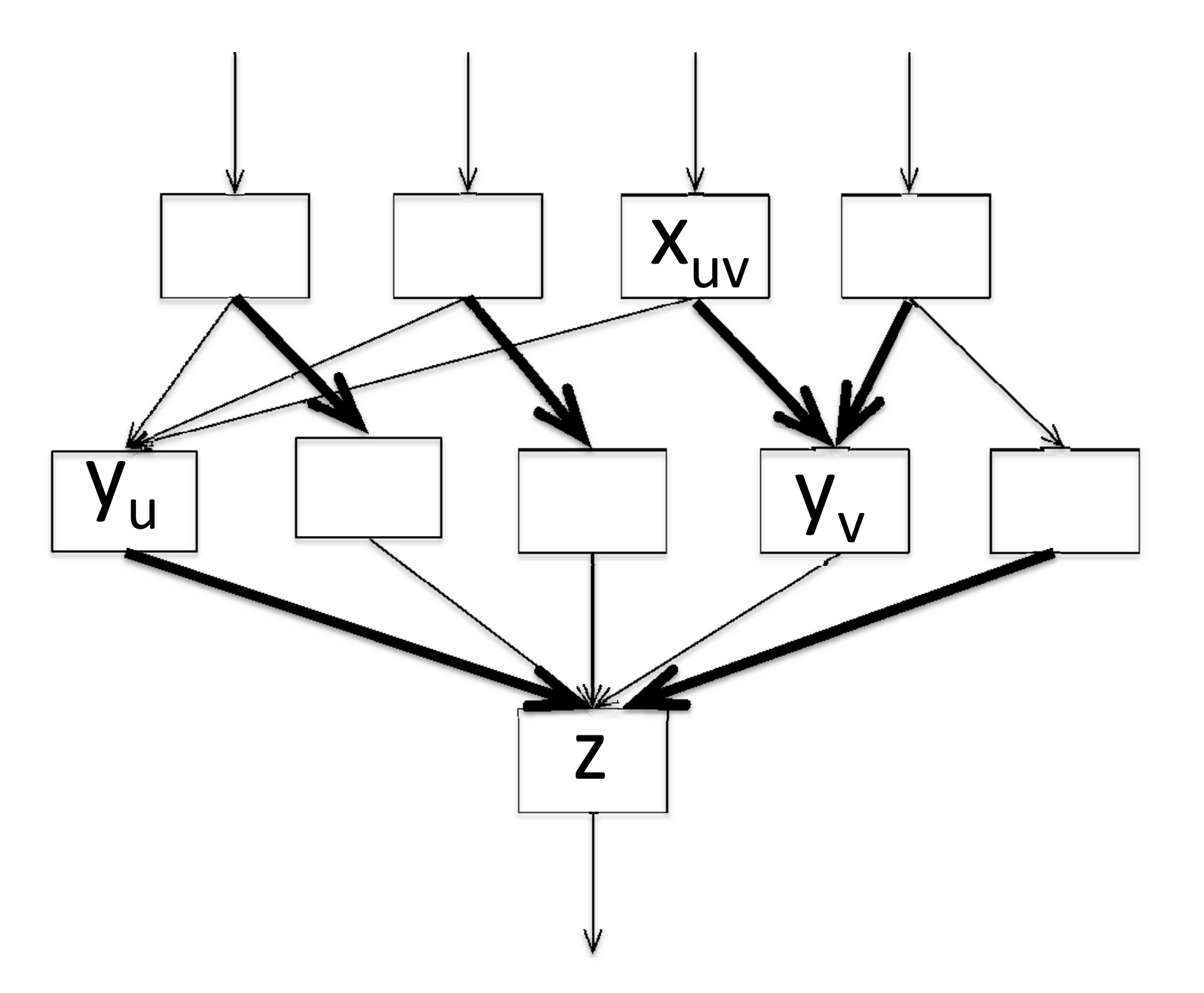}
  \end{center}
  \caption{Reduction from vertex cover, the dark edges show a solution 
  with cost $|E'| + K$, $K = $ size of a vertex cover in $G'$}
  \label{fig:vc}
\end{figure}

Now we define the requirement list for each module in $W$.
For each $x_{uv}$, $L_{uv} = \{(0, 1)\}$, i.e.
the requirement for $x_{uv}$ is any single outgoing edge.
For each $y_{v}$, $L_{v} = \{(d_v, 0), (0, 1)\}$, where $d_v$ is the
degree of the $v$ in $G'$. Hence the requirement of
the vertex $y_v$ is either all of its incoming edges, or
a single outgoing edge. For vertex $z$, $L_z = \{(1, 0)\}$, i.e.
hiding any incoming edge suffices. 


\begin{lemma}\label{lem:vc}
The vertex cover instance $G'$ has a solution of size $K$ if and only if the \secureview instance $W$
has a solution of cost $m' + K$, where $m' = |E'|$ is the number of edges in $G'$.
\end{lemma}
\begin{proof}
Let $S \subseteq V'$ be a vertex cover of $G'$ of size $K$. 
We create a create a set of hidden edges $\widebar{V}$ for the \secureview
problem as follows:
for each $v \in S$, add $(y_{v}, z)$ to $\widebar{V}$. Further, for each $x_{uv}$, if $u \notin S$, add
$(x_{uv}, y_{u})$ to $\widebar{V}$, otherwise add $(x_{uv}, y_{v})$. 
For this choice of $\widebar{V}$, we claim that $V$ is safe set of attributes for $W$.
\par
Clearly, the requirement is satisfied for each $x_{uv}$, since one outgoing edge is hidden;
the same holds for all $y_v$ such that $v \in S$. 
Assuming $E'$ to be non-empty, any vertex cover is of size at least one.
Hence at least one incoming edge to $z$ is hidden. 
Finally, for every $y_{v}$ such that $v \notin S$, all its incoming edges are hidden;
if not, $S$ is not a vertex cover. 
Hence $\widebar{V}$ satisfies the requirement of all modules in $W$.
Since we hide exactly one outgoing edge from all $x_{uv}$,
and exactly one outgoing edge from all $y_{v}$ where $v \in S$,
the cost of the solution is $m' + K$.

Now assume that we have a solution $\widebar{V} \subseteq A$ 
of the \secureview instance with cost $|\widebar{V}| = K'$.
We can assume, wlog, that for each $x_{uv}$ exactly 
one outgoing edge is included in $\widebar{V}$;
if both $(x_{uv}, y_{u})$ and $(x_{uv}, y_{v})$ are in $\widebar{V}$, 
we arbitrarily select one of $u$ or $v$, 
say $u$, and replace the edge $(x_{uv}, y_{u})$ in $\widebar{V}$ with the edge $(y_{u}, z)$
to get another feasible solution without increasing the cost.
We claim that the set $S \subseteq V'$ of vertices $v$ 
such that $(y_{v}, z) \in \widebar{V}$ forms a vertex cover.
For any edge $(u, v) \in E'$,
if $(x_{uv}, y_{u}) \notin \widebar{V}$, then $(y_{u}, z) \in \widebar{V}$
to satisfy the requirement of $y_u$, and therefore $u \in S$;
otherwise, $v \in S$ by the same argument.
Hence $S$ is a vertex cover.
Since each vertex $x_{uv}$ has exactly one outgoing edge in $\widebar{V}$,
$|S| = K' - m'$. 
\end{proof}

To complete the proof, note that 
if $G'$ were a cubic graph, i.e. the degree of each vertex is at most 3,
then the size of any vertex cover $K \geq m'/3$.
It is known that vertex cover in cubic graphs is APX-hard \cite{ASK97};
hence so is the \secureview problem with cardinality constraints and no data sharing.
An exactly identical reduction shows that the \secureview problem 
with set constraints and no data sharing is APX-hard as well.

\section{Proofs from Section~5} 
\label{sec:app_proofs_public-module}

\subsection{Proof of Lemma~\ref{lem:main-public}}\label{sec:main-public}

In this section we state and prove Lemma~\ref{lem:main-public}, which uses the same notations as in
Lemma~\ref{lem:main-private}. Here again $m_i$ is a fixed private module.

\begin{lemma}\label{lem:main-public}
If $m_j, j \in [K+1, n]$ is a public module such that
$m_j \neq g_j$ in the proof of
Lemma~\ref{lem:main-private}, then $(I_j \cup O_j) \cap \widebar{V_i} \neq
\phi$, and therefore $m_j$ will be hidden.
\end{lemma}

\begin{proof}
We will show that if $(I_j \cup O_j) \cap
\widebar{V_i} = \phi$, then $m_j = g_j$.
\par
Recall that we defined two tuples $\tup{p}, \tup{q}$ over attributes $I_i \cup O_i$
in the proof of
Lemma~\ref{lem:main-private} and argued that $p[a] = q[a]$ for all
the attributes $a \in (I_i \cup O_i) \cap V_i$. Hence if $p[a] \neq
q[a]$, then $a \in \widebar{V_i}$, i.e. $a$ is hidden. From the
definition of $\Flip$ it follows that, when $(I_j \cup O_j) \cap
\widebar{V_i} = \phi$, for an input $\tup{u}$ to $m_j$
$\Flip_{\tup{p, q}}(\tup{u}) = \tup{u}$. Similarly, $\Flip_{\tup{p,
q}}(\tup{v}) = \tup{v}$, where $\tup{v} = m_j(\tup{u})$. Hence for
any input $\tup{u}$ to $m_j$,
\begin{eqnarray*}
g_j(\tup{u}) & = & \Flip_{m_j, \tup{p, q}}(\tup{u})\\
& = & \Flip_{\tup{p, q}}(m_j(\Flip_{\tup{p, q}}(\tup{u})))\\
& = & \Flip_{\tup{p, q}}(m_j(\tup{u}))\\
& = & \Flip_{\tup{p, q}}(\tup{v})\\
& = & \tup{v}\\
& = & m_j(\tup{u}).
\end{eqnarray*}
Since this is true for all input $\tup{u}$ to $m_j$, $m_j = g_j$ holds.
\end{proof}

\subsection{Bounded Data Sharing}\label{sec:bounded-public}

In Section~\ref{sec:bounded-private} we showed that the \secureview problem with cardinality or set constraints
has a $\gamma+1$-approximation where every data item in $A$
can be fed as input to at most $\gamma$ modules.
This implies that without
any data sharing (when $\gamma = 1$), \secureview with cardinality or set
constraints had a 2-approximation. In the following theorem 
we show that
in arbitrary networks, this problem is $\Omega(\log n)$-hard to approximate.

\begin{theorem}\label{thm:bounded-public}
\textbf{Bounded Data sharing  (general workflows):~~}
Unless ${\rm NP} \subseteq {\rm DTIME}(n^{O(\log \log n)})$, the \secureview problem with cardinality constraints
without data sharing
 in general workflows is
 $\Omega(\log n)$-hard to approximate even if the maximum size of the requirement lists is 1 and the individual
requirements are bounded by 1.
\end{theorem}

\begin{proof}
The reduction will again be from the set cover problem.
An instance of the set cover problem consists of
an input universe ${\bf U} = \{u_1, u_2, \ldots, u_{n'}\}$, 
and a set of its subsets
${\bf S} = \{S_1, S_2, \ldots, S_{m'}\}$, i.e.
each $S_i\subseteq U$. The goal is to find a set of subsets
${\bf T}\subseteq {\bf S}$ of minimum size (i.e. $|{\bf T}|$
is minimized) subject to the constraint that 
$\cup_{S_i\in {\bf T}} S_i = {\bf U}$.
\par
Given an instance of the set-cover problem, we construct a workflow $W$
as follow: (i) 
we create a public module for every element in $U$,
(ii) 
we create a private module for every set in ${\bf S}$,
(iii) we add an edge $(S_i, u_j)$ with data item $b_{ij}$ if and only if $u_j \in S_i$.
Every set $S_i$ has an incoming edge with data item $a_i$ 
(initial input data)
and every element $u_j$ has an outgoing edge with data item 
$b_j$ (final output data).  
The cost of hiding every edge is 0 and the cost of privatizing every set node
$S_i$ is 1. The requirement list of every private module $u_j$ is $L_{j} = \{(1, 0)\}$,
i.e., for every such module one of the incoming edges must be chosen.
\par
It is easy to verify that the set cover has a solution of size $K$ if and only if the
\secureview problem has a solution of cost $K$. 
Since set-cover is known to be $\Omega(\log n')$-hard to approximate, $m'$ is polynomial in $n'$ in the construction in \cite{LY94}, 
and the total number of nodes in $W$, $n = O(n' + m')$, 
this problem has
the same hardness of approximation as set cover, i.e. the problem is $\Omega(\log n)$-hard
to approximate.
\end{proof}
\subsection{Cardinality Constraints}\label{sec:cardinality-public}
Here we show that the \secureview problem with cardinality constraints 
is $\Omega(2^{\log^{1-\gamma}n})$-hard to approximate.
This is in contrast with the $O(\log n)$-approximation obtained for this problem
in all-private workflows (see Theorem~\ref{thm:cardinality-private}).

\begin{theorem}\label{thm:cardinality-public}
\textbf{Cardinality Constraints  (general workflows):~~}
Unless ${\rm NP} \subseteq {\rm DTIME}(n^{{\rm polylog~} n})$,
the \secureview problem with cardinality constraints in general workflows is
$\Omega(2^{\log^{1-\gamma}n})$-hard to approximate
for all constant $\gamma > 0$ even if the maximum size of the requirement lists is 1 and the individual
requirements are bounded by 1.
\end{theorem}

The hardness result in 
Theorem~\ref{thm:cardinality-public} is obtained by a reduction from the 
\emph{minimum label cover} problem~\cite{ABS+93}. An instance of
the minimum label cover problem consists of a bipartite graph 
$H = (U, U', E_H)$, a label set $L$,
and a non-empty relation $R_{uw} \subseteq L \times L$ for each edge
$(u, w) \in E_H$. A feasible solution is a label assignment to the 
vertices, $A : U \cup U'  \rightarrow 2^{L}$, 
such that for each
edge $(u, w)$, there exist $\ell_1 \in A(u), \ell_2 \in A(w)$ 
such that $(\ell_1, \ell_2) \in R_{uw}$.
The objective is to find a feasible solution that minimizes
$\sum_{u \in U \cup U'}|A(u)|$.
If $N = |U|+|U'|$, the label cover problem is known to be
$|L|^\epsilon$-hard to approximate
for some constant $\epsilon > 0$, as well as, 
$\Omega(2^{\log^{1-\gamma}N})$-hard to approximate
for all constant $\gamma > 0$, unless 
${\rm NP} \subseteq {\rm DTIME}(n^{{\rm polylog~} n})$~\cite{ABS+93, Raz98}. 
\par
Given an instance of the label cover problem as defined above, 
we create an instance 
of the \secureview problem 
by constructing a workflow $W$ (refer to Figure~\ref{fig:lc2}).
We will show that the label cover instance $H$ has a solution of  cost $K$
if and only if the \secureview instance $W$ has a solution with the same cost $K$.
Further, in our reduction, the number of modules $n$ in the workflow $W$ will be $O(N^2)$.
Hence the \secureview problem with cardinality constraints in general workflows will be
$\Omega(2^{\log^{1-\gamma}n})$-hard to approximate 
for all constant $\gamma > 0$ under the same complexity
assumption which proves Theorem~\ref{thm:cardinality-public}.

\begin{figure}[h]
  \begin{center}
    \includegraphics[scale=.35]{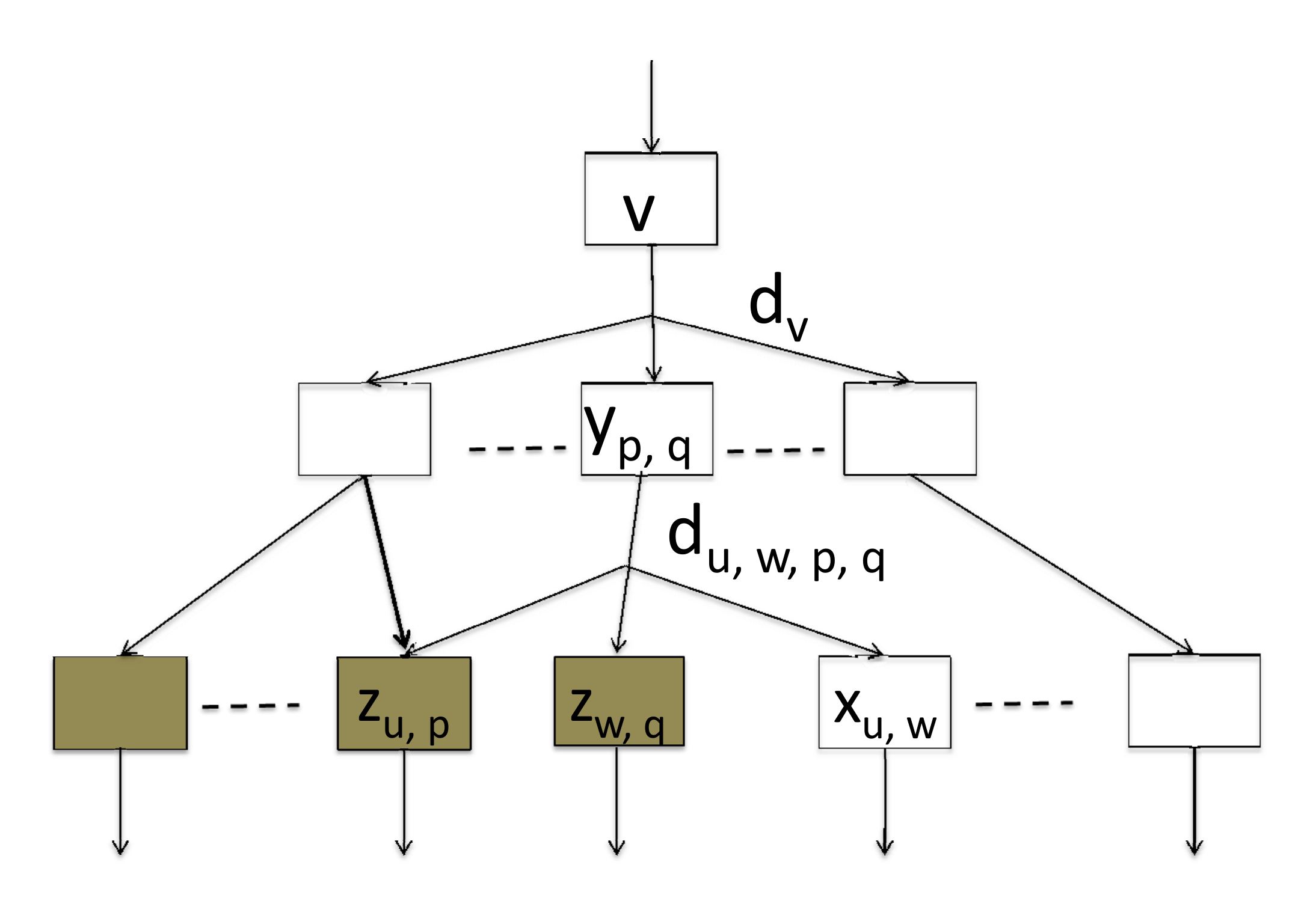}
  \end{center}
  \caption{Reduction from label cover: $(p, q) \in R_{uw}$, 
  the public modules are darken, all public modules have unit cost, all data have zero cost,
  the names of the data never hidden are omitted for simplicity}
  \label{fig:lc2}
\end{figure}

\paragraph{Construction}
First we describe the modules in $W$.
For each edge $(u, w) \in E_H$,
there is a private module $x_{u, w}$ in $W$. 
For every pair of labels $(\ell_1, \ell_2)$, there is a private module
$y_{\ell_1, \ell_2}$. There are public module $z_{u, \ell}$ for every 
$u \in U \cup U', \ell \in L$. 
In addition, $W$ has another 
private module $v$. 
\par
As shown in Figure~\ref{fig:lc2}, there are the following
types of edges and data items in $W$: 
(i) an incoming  single edge to $v$ carrying initial input data item $d_s$, $\cost(d_s) = 0$, 
(ii) an edge from $v$ to every node $y_{\ell_1, \ell_2}$, each such edge carries
the same data $d_v$ produced by $v$, $\cost(d_v) = 0$,
(iii) for every $(u, w) \in E_H$, for $(\ell_1, \ell_2) \in R_{u, w}$,
there is an edge from $y_{\ell_1, \ell_2}$ to $x_{u, w}$ carrying data $d_{u, w, \ell_1, \ell_2}$,
and $\cost(d_{u, w, \ell_1, \ell_2}) = 0$,
(iv) further, every such data $d_{u, w, \ell_1, \ell_2}$ produced by $y_{\ell_1, \ell_2}$
is also fed to both $z_{u, \ell_1}$ and $z_{v, \ell_2}$.
(v) All $y_{\ell_1, \ell_2}$ and all $z_{u, \ell}$ have an outgoing edge 
carrying data (final output data items) 
$d_{\ell_1, \ell_2}$ and $d_{u, \ell_1}$ respectively; $\cost(d_{\ell_1, \ell_2}) = 0$ and $\cost(d_{u, \ell_1}) = 0$.
\par
Privacy requirement of $v$ is $\{(0, 1)\}$, i.e., $v$ has to always choose
the output data $d_v$ to satisfy its requirement. Privacy requirement of every $y_{\ell_1, \ell_2}$
is $\{(1, 0)\}$, i.e. choosing $d_v$ satisfies the requirement for all such $y_{\ell_1, \ell_2}$-s.
Requirements of every $x_{u, w}$ is $(1, 0)$, so one of the data items $d_{u, w, \ell_1, \ell_2}$,
$(\ell_1, \ell_2) \in R_{u, w}$ must be chosen. 
\par
All modules except 
$z_{u, \ell}$-s are private. Cost of
privatizing $z_{u, \ell}$, $\pvt(z_{u, \ell}) = 1$. 
In the above reduction, all data items have cost 0 and the cost of a solution entirely
comes from privatizing the public modules $z_{u, \ell}$-s. 
\par
In this reduction, maximum list size and maximum magnitude
of any cardinality requirement are both bounded by 1. In the label cover instance it is known that number of labels
$L \leq $ number of vertices $N$, therefore the total number of modules in $W$ = $O(L^2 + LN + N^2) = O(N^2)$.
The following lemma 
proves the correctness of the reduction.

\noindent
\begin{lemma}\label{lem:lc-general}
The label cover instance $H$ has a solution of cost $K$ iff 
the \secureview instance $W$ has a solution of cost $K$.
\end{lemma}
\begin{proof}
Let $A: U \cup U' \rightarrow 2^L$ be a solution of the label cover instance $H$
with total cost $K = \sum_{u \in U \cup U'} |A(u)|$. 
We create a solution $\widebar{V}$ for the \secureview instance $W$ as follows:
First, add $d_v$ to the hidden attribute subset $\widebar{V}$.
This satisfies the requirement of $v$ and all $y_{\ell_1, \ell_2}$
without privatizing any public modules. So we need to satisfy the requirements of $x_{u, w}$.
\par
Since $A$ is a valid 
label cover solution, for every edge $(u, w) \in E$, there is a label pair $(\ell_1, \ell_2) \in R_{uw}$
such that $\ell_1 \in A(u)$ and $\ell_2 \in A(w)$. For every $x_{u, w}$, add such $d_{u, w, \ell_1, \ell_2}$
to  $\widebar{V}$. For each $u \in U \cup U'$, and $\ell \in L$, add $z_{u, \ell}$ to the privatized public modules $P$.
iff $\ell \in A(u)$. It is easy to check that $(V, P)$ is safe for $W$. If $d_{u, w, \ell_1, \ell_2}$ is hidden,
both $z_{u, \ell_1}$ and $z_{v, \ell_2}$ are added to $P$. Since $\cost(\widebar{V}) = 0$, cost of the solution
is $\pvt(\widebar{P}) = K = \sum_{u \in U \cup U'} |A(u)|$.
\par
Conversely, let $(V, P)$ be a safe solution of the \secureview instance $W$, where $K = \pvt(\widebar{P})$.
For each $u \in U \cup U'$, we define $A(u) = \{\ell | z_{u, \ell} \in P\}$.
Clearly, $\sum_{u \in U \cup U'} |A(u)| = K$. 
For each $(u, w) \in E$, the requirement of $x_{u, w}$ is satisfied by $\widebar{V}$;
hence there exist $\ell_1, \ell_2 \in L$ such that $d_{u, w, \ell_1, \ell_2} \in \widebar{V}$.
Therefore both $z_{u, \ell_1}, z_{w, \ell_2} \in P$.
This implies that for each edge $(u, w) \in E_H$, there exist $\ell_1 \in A(u)$
and $\ell_2 \in A(w)$, where $(\ell_1, \ell_2) \in R_{uw}$, thereby proving
feasibility. 
\end{proof}

\subsection{Set-Constraints}  
\label{sec:set-public}
We modify the LP given in Section~\ref{sec:set-private}
and give  an $\ell_{\max}$-approximation algorithm for the set-constraints version
in general workflows.
As before, for an attribute $b \in A$, $x_b = 1$ if and only if $b$ is hidden (in the final solution $b \in \widebar{V}$).
In addition, for a public module $m_i, i \in [K+1, n]$, $w_i = 1$ if and only if $m_i$
is hidden (in the final solution $m_i \in \widebar{P}$). 
The algorithm rounds the solution given by LP relaxation of the following integer program.
The new condition introduced is (\ref{equn:pub-IP2-4}), which says that,
if any input or output attribute of a public module $m_i$ is included in $\widebar{V}$,
$m_i$ must be hidden. Further, (\ref{equn:pub-IP2-1})
is needed only for the private modules ($m_i$ such that $i \in [1, K]$).
For simplicity we denote $\cost(b) = c_b$
and $\pvt(m_i) = \pvt_i$. 

\begin{center}
Minimize $\sum_{b \in A} \cost_b x_b + \sum_{i \in [K+1, n]} \pvt_i w_i \quad$ subject to 
\begin{align}
& \sum_{j = 1}^{\ell_i} r_{ij} \geq 1 & \forall i \in [1, K] \label{equn:pub-IP2-1}\\
& x_b \geq r_{ij} & \forall b \in I_{ij} \cup O_{ij},~\forall i \in [1, K] \label{equn:pub-IP2-2}\\%
& w_i \geq x_b & \forall b \in I_{i} \cup O_{i},~\forall i \in [K+1, n] \label{equn:pub-IP2-4}\\%
& x_b, r_{i, j}, w_i \in \{0, 1\} &\label{equn:pub-IP2-3}
\end{align}
\end{center}
The LP relaxation is obtained by changing Constraint~(\ref{equn:pub-IP2-3}) to
\begin{equation}
x_b, r_{ij}, w_i \in [0, 1].  \label{equn:pub-LP2-3}
\end{equation}

The rounding algorithm outputs $\widebar{V} = \{b : x_b \geq 1/\ell_{\max}\}$
and $\widebar{P} = \{m_i : b \in \widebar{V} \text{ for some } b \in I_i \cup O_i\}$.
\par 
Since the maximum size of a requirement list is $\ell_{\max} = \max_{i = 1}^n \ell_i$, 
for each $i$, there exists a $j$ such that in the solution of the LP,
$r_{ij} \geq 1/\ell_i \geq 1/\ell_{\max}$ (from (\ref{equn:pub-IP2-1})). 
Hence there exists at least one $j \in [1, \ell_i]$ such that $I_{ij} \cup O_{ij}  \subseteq \widebar{V}$.
Further, a public module $m_i$, $i \in [K+1, n]$ is hidden (i.e. included to $\widebar{V}$) only if
there is an attribute $b \in I_i \cup O_i$ which is included to $\widebar{V}$.
Therefore from (\ref{equn:pub-IP2-4}), for all $m_i \in \widebar{P}$, $w_i \geq \frac{1}{\ell_{\max}}$.
Since both $\cost(\widebar{V})$ and $\pvt(\widebar{P})$ are most $\ell_{\max}$ 
times the cost of the respective cost in the LP solution, 
this rounding algorithm gives an $\ell_{\max}$-approximation.

\end{document}